\documentclass{article}

\usepackage[preprint]{neurips_2026}

\usepackage{todonotes}
\usepackage{comment}

\usepackage[utf8]{inputenc}
\usepackage[T1]{fontenc}
\usepackage{microtype}

\usepackage[main=UKenglish]{babel}
\usepackage[strict]{csquotes}

\usepackage{paralist}
\usepackage{enumitem}

\usepackage{wrapfig}
\usepackage[caption=false]{subfig}

\usepackage{pdflscape}

\usepackage{tikz}
\usetikzlibrary{arrows.meta, positioning, shadows}

\usepackage{pgfplots}
\pgfplotsset{compat=1.18}

\usepackage[export]{adjustbox}

\usepackage{amsfonts}
\usepackage{amssymb}
\usepackage{amsmath}
\usepackage{amsthm}
\usepackage{bbm}
\usepackage{aliascnt}

\usepackage{pifont}

\theoremstyle{plain}
\newtheorem{theorem}{Theorem}[section]

\newaliascnt{proposition}{theorem}
\newtheorem{proposition}[proposition]{Proposition}
\aliascntresetthe{proposition}

\newaliascnt{lemma}{theorem}

\aliascntresetthe{lemma}

\newaliascnt{fact}{theorem}

\aliascntresetthe{fact}

\newaliascnt{property}{theorem}

\aliascntresetthe{property}

\theoremstyle{definition}
\newaliascnt{remark}{theorem}
\newtheorem{remark}[remark]{Remark}
\aliascntresetthe{remark}

\newaliascnt{definition}{theorem}
\newtheorem{definition}[definition]{Definition}
\aliascntresetthe{definition}

\newaliascnt{notation}{theorem}

\aliascntresetthe{notation}

\newaliascnt{construction}{theorem}

\aliascntresetthe{construction}

\newaliascnt{example}{theorem}
\newtheorem{example}[example]{Example}
\aliascntresetthe{example}

\newaliascnt{counterexample}{theorem}

\aliascntresetthe{counterexample}

\newaliascnt{nonexample}{theorem}

\aliascntresetthe{nonexample}

\usepackage{derivative}
\derivset{\pdv}[fun=true] %,switch-*=true]

\usepackage{nicefrac}

\usepackage{stmaryrd}
\usepackage{siunitx}
\DeclareSIUnit\feet{ft}

\usepackage{caption}
\usepackage{tabularray}
\UseTblrLibrary{amsmath, booktabs, siunitx}

\NewTblrTheme{paper}{
  \SetTblrStyle{caption}{font=\small}
  \SetTblrStyle{caption-tag}{font=\small\bfseries}
  \SetTblrStyle{note}{font=\footnotesize}
}

\usepackage[table]{xcolor}
\usepackage{hhline}

\input{Definitions}
\usepackage[all]{nowidow}
\usepackage[lastparline]{impnattypo}

\usepackage{centernot}

\usepackage[thicklines]{cancel}

\usepackage{tikz}

\usetikzlibrary{cd, arrows}
\usetikzlibrary{decorations.pathreplacing}

\usepackage{fonttable}
\usepackage{wasysym}
\usepackage{expl3}
\usepackage{xparse, xpatch}

\usepackage{tabularx, multirow}
\usepackage{stackrel, stackengine, old-arrows}

\usepackage{ebproof}

\usepackage{mathrsfs, amsfonts, stmaryrd, mathtools}
\usepackage{amsopn}

\usepackage{quiver}

\usepackage{xspace}

\newenvironment{eqalign*}{\begin{equation*}\begin{aligned}}{\end{aligned}\end{equation*}}

\tikzset{
  relation/.style={
    draw=none,
    every to/.append style={
      edge node={node [sloped, allow upside down, auto=false]{$#1$}}}
  }
}

\tikzcdset{
  mono/.code={
    \pgfsetarrows{tikzcd to reversed-tikzcd to}
  }
}
\tikzcdset{
  epi/.code={
    \pgfsetarrowsend{tikzcd double to}
  }
}
\tikzcdset{
  into/.code={
    \pgfsetarrows{tikzcd right hook-tikzcd to}
  }
}
\tikzcdset{
  twocell/.style={Rightarrow, shorten >= 3ex, shorten <= 3ex}
}

\newcommand{\Overset}[2]{%
  \mathop{#2}\limits^{\vbox to -.1ex{%
  \kern -1.8ex\hbox{$#1$}\vss}}%
}
\newcommand{\Underset}[2]{%
  \mathop{#2}\limits_{\vbox to .1ex{%
  \kern -.6ex\hbox{$#1$}\vss}}%
}

\mathchardef\dash="2D

\renewcommand{\exp}{\operatorname{exp}}

\newcommand{\word}[1]{\quad\text{\underline{#1}}\quad}

\newcommand{\impl}{\word{implies}}

\newcommand{\biginf}{\bigwedge}

\newcommand{\psum}[2][p]{\mathchoice{{\bigoplus_{#2}}^{#1}}{\bigoplus_{#2}^{#1}}{\bigoplus_{#2}^{#1}}{\bigoplus_{#2}^{#1}}}

\newcommand{\phsum}[2][p]{\psum[-#1]{#2}}

\newcommand{\Reals}{[-\infty,+\infty]}
\newcommand{\ExtReals}{\Reals}
\newcommand{\PosReals}{[0,\infty]}

\newcommand{\MulReals}{\PosReals_{\Mtensor}}
\newcommand{\AddReals}{\ExtReals_{\Atensor}}

\newcommand{\mulimp}{\multimap}

\newcommand{\One}{\literal{0}}
\newcommand{\true}{{\bf true}}
\newcommand{\false}{{\bf false}}

\newcommand{\literal}[1]{\mathbf{#1}}

\newcommand{\sem}[2][]{\llbracket #2 \rrbracket_{#1}}

\newcommand{\red}[1]{\textcolor{red}{#1}}

\newcommand{\pQLL}[1][p]{\ensuremath{#1\text{QLL}}\xspace}
\newcommand{\pFor}[1][p]{\mathbin{\lor^{#1}}}
\newcommand{\pFand}[1][p]{\mathbin{\land^{#1}}}
\newcommand{\For}{\pFor[]}
\newcommand{\Fand}{\pFand[]}
\newcommand{\Fnot}[1]{\neg #1}

\newcommand{\pSor}[1][p]{\color{red}{\pAor}}
\newcommand{\pSand}[1][p]{\color{red}{\pAand}}

\newcommand{\pMor}[1][p]{\mathbin{+^{#1}}}

\newcommand{\Mtensor}{\mathbin{\times}}

\newcommand{\pAor}[1][p]{\mathbin{\cup^{#1}}}
\newcommand{\pAand}[1][p]{\mathbin{\cap^{#1}}}

\newcommand{\Atensor}{\mathbin{+}}

\newcommand{\booleize}{\operatorname{!}}

\newcommand{\grammareq}{\ \Coloneqq \ }
\newcommand{\grammarsep}{\ | \  }

\newcommand*{\oracle}[2]{\red{a}}

\renewcommand{\epsilon}{\varepsilon}

\usepackage{acronym}

\acrodef{ML}{machine learning}
\acrodef{NeSy}{neuro-symbolic}
\acrodef{DL2}{Deep Learning with Differentiable Logics}
\acrodef{PGD}{Projected Gradient Descent}
\acrodef{FGSM}{Fast Gradient Sign Method}
\acrodef{VNNCOMP}[VNN-COMP]{Neural Network Verification Competition}
\acrodef{FOL}{first-order logic}
\acrodef{STL}{Signal Temporal Logic}
\acrodef{LTL}{Linear Temporal Logic}
\acrodef{SBR}{Semantic-based Regularisation}
\acrodef{SPL}{Semantic Probabilistic Layer}
\acrodef{DRL}{Disjunctive Refinement Layer}
\acrodef{PAL}{Probabilistic Algebraic Layer}
\acrodef{LDL}{Logic of Differentiable Logics}
\acrodef{RMSE}{Root Mean Squared Error}
\acrodef{PAcc}{Prediction Accuracy}
\acrodef{CAcc}{Constraint Accuracy}
\acrodef{CSec}{Constraint Security}
\acrodef{CSat}{Verified Constraint Satisfaction}
\acrodef{GTSRB}{German Traffic Sign Recognition Benchmark}
\acrodef{QLL}{Quantitative Linear Logic}
\acrodef{SRA}{Soft Residuated Algebra}
\acrodef{SMT}{Satisfiability Modulo Theories}
\acrodef{LL}{Linear Logic}

\makeatletter
\AtBeginDocument
 {
	 \def\ltx@label#1{\cref@label{#1}}%add braces
	 \def\label@in@display@noarg#1{\cref@old@label@in@display{#1}}%remove braces
\def\label@in@mmeasure@noarg#1{%
		\begingroup%
			\measuring@false%
			\cref@old@label@in@display{#1}%remove braces for multline, see https://tex.stackexchange.com/q/737204/2388
		\endgroup}%
 } %
\makeatother

\usepackage{hyperref}
\usepackage[capitalize, nameinlink]{cleveref}

\crefname{diagram}{Diagram}{Diagrams}
\crefname{rule}{Rule}{Rules}
\crefname{proof}{Derivation}{Derivations}

\usepackage{url}

\newcommand{\dl}[2]{\llbracket #1\rrbracket_{\mathrm{#2}}}

\DeclareMathOperator*{\argmin}{arg\,min}

\DeclareMathOperator*{\expected}{\mathbb{E}}

\usepackage{bm}
\renewcommand{\vec}{\bm}

\title{
	Quantitative Linear Logic for Neuro-Symbolic Learning and Verification
}

\author{%
	Thomas Flinkow\\
	Maynooth University, Maynooth, IE\\
	\texttt{thomas.flinkow@mu.ie} \\
	\And
	Ekaterina Komendantskaya\\
	University of Southampton, Southampton, UK\\
	\texttt{ek1u23@soton.ac.uk} \\
	\And
	Matteo Capucci\\
	University of Strathclyde, Glasgow, UK\\
	Independent Researcher, Modena, IT\\
	\texttt{matteo.capucci@gmail.com} \\
    \And
	Rosemary Monahan\\
	Maynooth University, Maynooth, IE\\
	\texttt{rosemary.monahan@mu.ie} \\
}

\begin{document}

\maketitle

\begin{abstract}
    Differentiable Logics are deployed in neuro-symbolic learning tasks as a way of embedding logical constraints in the training objective of neural networks.
    A differentiable logic consists of a syntax to write logical properties and a semantics to interpret them as real-valued functions to be folded in the loss function.
    A defining trade-off of the field is that between logical properties of the connectives, and analytic concerns for the semantics, with both aspects being relevant in applications.
    At one extreme we find fuzzy logics, that have well-established algebraic and proof-theoretic foundations, and at the other ad-hoc differentiable logics like Fischer's DL2, conceived for deep learning applications.
    However, no satisfactory foundation has emerged yet.
    We propose a resolution to this long-standing tension via a novel logic, \ac{QLL}, with foundational ambitions.
    Our design is driven by \emph{naturality}---the idea that, since logical constraints are translated to losses, the semantics of the connectives should be pertinent operations used in ML practice (that is, sum and log-sum-exp) on additive quantities (like logits).
    We then judge the result on two aspects: \emph{logical adequacy}---that they satisfy most of the standard logical laws of Linear Logic; and \emph{empirical effectiveness}---test-time performance (as measured by adversarial attacks) is well-correlated to the actual verification of the logical constraints (as measured by off-the-shelf neural network verifiers), which makes \ac{QLL} stand out among SoTA techniques.
\end{abstract}

\section{Introduction}
\label{sec:introduction}

Neuro-symbolic machine learning (NeSy ML) is a term that refers to methods that combine data-driven optimisation (``the neural") with logical reasoning (``the symbolic").
The slogan \enquote{The future is neuro-symbolic}~\citep{Belle_Marcus_2026} is evidenced by a range of notable successes, from the  flagship project on Alpha Geometry~\citep{trinhSolvingOlympiadGeometry2024} to numerous compendia of practical neuro-symbolic methods~\citep{NeSY23}.

Neuro-symbolic tasks occur in every area of machine learning.
In cyber-physical systems, for example,
models such as the famous ACAS Xu~\citep{katzReluplexEfficientSMT2017,katzMarabouFrameworkVerification2019} are obtained by training a neural aircraft controller to avoid mid-air collision in a way that is optimised relative to the prior data, but at the same time, ensuring that certain unsafe decisions never occur by mistake.
These unsafe scenarios exist as \emph{prior knowledge} of the engineers, and can be formally specified and verified as logical constraints that restrict the neural network output for certain regions of the input space.
%, prior to the deployment of such controllers.
In computer vision, logical constraints arise when classes of images naturally organize hierarchically.
% For example, in the \ac{GTSRB} data set~\citep{stallkampGermanTrafficSign2011}, road signs can be grouped into \emph{superclasses}: \enquote{speed limits}, \enquote{warning signs}, etc.
In Fashion-MNIST~\citep{xiaoFashionMNISTNovelImage2017}, `footwear' and `clothing' form \emph{superclasses}.
Then an \emph{`anti-misclassification'}  constraint may then assert that misclassifications inside the same superclass are preferable to superclass misclassifications.
Finally, \emph{(adversarial) robustness}~\citep{szegedyIntriguingPropertiesNeural2014}, a problem that affects all areas of ML from computer vision to natural language processing, arguably also boils down to a logical constraint~\citep{casadioNeuralNetworkRobustness2022} which forces the model's predictions to exceed a certain threshold (\emph{classification robustness}) or to remain robust to small perturbations of the dataset.
The requirement for classification to remain the same for all data within small $\varepsilon$-balls of the data set samples is a requirement that is, in essence, extraneous to the question of optimising the model for the data as given.

 \begin{wrapfigure}[11]{R}{0.5\linewidth}
    \centering
    \subfloat[]{%
        \includegraphics[width=.175\linewidth]{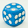}%
    }\hfil
    %\subfloat[]{%
    %    \includegraphics[width=.15\linewidth]{figures/066.png}%
    %}\hfil
    \subfloat[]{%
        \includegraphics[width=.175\linewidth]{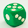}%
    }\hfil
    \subfloat[]{%
        \includegraphics[width=.175\linewidth]{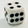}%
    }\hfil
    \subfloat[]{%
        \includegraphics[width=.175\linewidth]{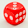}%
    }
    \caption{Images used as a neuro-symbolic multi-label classification training benchmark.
    The classifier is trained to recognise visible faces of a die.
    }
    \label{fig:dice-images}
\end{wrapfigure}
Then these three major categories of the NeSy ML tasks (reflecting prior knowledge of the physical world, encoding class relations, and robustness) can be mixed in a variety of ways.
For example, consider a computer vision task in which we train a multi-label classifier $f\colon\real^{3\times 28\times 28}\to\real^6$ to recognise the pips on the visible faces of a die from a $28\times 28$ RGB image.
We have some \emph{prior knowledge} that some combinations of visible faces are physically impossible.
We can form \emph{superclasses} of permissible combinations (e.g. $(2,4,6)$, $(4,2,6)$, $(6,2,4)$, etc. form one superclass), and we can specify a logical constraint that says that within \emph{$\varepsilon$-vicinity} of each image in the data set, the classification will only fall within a permissible superclass.
Thus, in addition to using visual data (see~\cref{fig:dice-images}), we want to incorporate prior knowledge about relationships between visible faces during training, i.e. we want to teach the classifier to predict only physically possible combinations of pips.

It has been noticed long ago (most famously, but not exclusively, in~\citep{szegedyIntriguingPropertiesNeural2014}) that models trained only on data fail to satisfy logical constraints.
This gave rise to different flavors of \emph{property-driven training}~\citep{giunchigliaDeepLearningLogical2022}. %is inspired by the formal  verification of neural networks.
The goal of this group of methods is to establish a balance between prediction performance and satisfaction of the chosen logical constraint, irrespective of whether it is entailed by the data or not.
Eventually, we may want to guarantee (i.e. formally verify) that the model adheres to the property (see \emph{Related work} below).
Some safety-critical systems with neural network components are modelled, verified, and trained in that way.

\emph{Differentiable logic} (DL) is one specific methodology for property-driven training.
Unlike other methods surveyed e.g. by~\citet{giunchigliaDeepLearningLogical2022}, it does not make any assumption about the architecture of the model or the kind of data, and can be used to augment any optimisation task.
Given a property $\phi(\theta)$ that concerns an ML model $f_\theta$ with trainable parameters $\theta$, the goal is to define a loss function $\dl{\phi}{}(\theta)$ that can be optimised to find $\theta^*$ satisfying $\phi$.
In principle, $\phi$ can be an arbitrary propositional formula that includes logical connectives \emph{and} ($\land$), \emph{or} ($\lor$), \emph{negation} ($\neg$), etc.---a DL then interprets such connectives as operations over suitable subsets of the reals.

\acfi{DL2}~\citep{fischerDL2TrainingQuerying2019} was one of the first differentiable logics proposed specifically for ML.
`Robustness metrics' for \ac{STL} (a temporal extension of classical propositional logic) were proposed by~\citet{varnaiRobustnessMetricsLearning2020}, with connectives adapted specifically to the gradient-based optimisation, and further investigated in~\citep{chevallierGradSTLComprehensiveSignal2025}.
Fuzzy logics are also natural candidates, given their established logical tradition and well-understood proof theory---indeed the usage of well-studied \emph{$t$-norms}, which give semantics to fuzzy logic connectives, was investigated by~\citet{vankriekenAnalyzingDifferentiableFuzzy2022}.

\begin{table}[t]
    \definecolor{tblgood}{RGB}{200,232,191}
    \definecolor{tblbad}{RGB}{248,178,178}
    \definecolor{tblwarn}{RGB}{251,228,153}

    \newcommand{\yes}{\ding{51}}%
    \newcommand{\no}{\ding{55}}%
    \newcommand{\na}{n/a}

    \footnotesize

    \centering
    \begin{talltblr}
    [
        theme=paper,
        caption={Comparison of differentiable logics on logical and analytic features.},
        label={tab:dl-feature-comparison},
        note{a}={Defined only at syntax level, not as a semantic operation.},
        note{b}={These logics do not feature a varying family of connectives.},
        note{c}={We only remark on analytic properties of the subset of connectives chosen for training.},
        note{d}={Shadow-lifting defined on non-extremal values.}
    ]
    {
      width=\linewidth,
      %vlines, hlines,
      colspec={Q[c,cmd=\rotatebox{90},wd=.2cm]X[l]*{8}{Q[c]}},
      row{1}={guard, font=\bfseries, mode=text},
      row{2}={guard, mode=text},
      column{1}={font=\bfseries, valign=m},
      cell{3}{3}={bg=tblbad},
      cell{3}{4-6}={bg=tblgood},
      cell{3}{7-8}={bg=tblbad},
      %cell{3}{9}={bg=tblgood},
      cell{4}{7}={bg=tblwarn},
      cell{5}{7}={bg=tblbad},
      cell{5-6}{6}={bg=tblbad},
      cell{6}{7,8}={bg=tblbad},
      cell{7}{3,4,6}={bg=tblbad},
      %cell{8}{3,4,6}={bg=tblbad},
      cell{8}{6}={bg=tblbad},
      cell{9}{3-7}={bg=tblbad},
      cell{3}{9}={bg=tblgood},
      cell{4}{3,4,5,6,8,9}={bg=tblgood},
      cell{5}{3,4,5,6,8,9}={bg=tblgood},
      cell{6}{3,4,5,6,9}={bg=tblgood},
      cell{7}{5,7,8,9}={bg=tblgood},
      %cell{8}{5,7,8,9}={bg=tblgood},
      cell{8}{3,4,5,7}={bg=tblbad},
      cell{8}{6,8,9}={bg=tblgood},
      cell{9}{8,9}={bg=tblgood},
    }
    \toprule
      & \SetCell[r=2]{l,m} Feature & \SetCell[r=1,c=4]{c} Fuzzy & & & & \SetCell[r=2]{c,m} \acs{DL2} & \SetCell[r=2]{c,m} \acs{STL} & \SetCell[r=2]{c,m} \acs{QLL} \\
    \cmidrule{3-6}
      & & Gödel & {\L}ukasiewicz & Product & Yager & & & \\
    \midrule
      \SetCell[r=4, c=1]{l} Logical & Distinct non-linear/linear connectives & \no & \yes & \yes & \yes & \no & \no & \yes \\
      & Negation & \yes & \yes & \yes & \yes & \yes\TblrNote{a} & \yes & \yes \\
      & Implication & \yes & \yes & \yes & \yes & \no & \yes & \yes \\
      & Good algebraic properties & \yes & \yes & \yes & \yes & \no & \no & \yes \\
    \midrule
      \SetCell[r=3, c=1]{l} Analytic\TblrNote{c}  & Shadow-lifting connectives\TblrNote{d} & \no & \no & \yes & \no & \yes & \yes & \yes \\
      & $\land$ and $\lor$ converge to min and max & \no\TblrNote{b} & \no\TblrNote{b} & \no\TblrNote{b} & \yes & \no\TblrNote{b} & \yes & \yes\\
      & Non-vanishing gradients & \no & \no & \no & \no & \no & \yes & \yes \\
      %& \emph{Naturality} & \no & \no & \no & \no & \no & \yes \\
    \bottomrule
    \end{talltblr}
\end{table}

% \paragraph{Problem: conciliating logic and machine learning needs}
% (0) fuzzy logics have a good proof system based on LL, but are empirically inadequate because they are not properly differentiable. most importantly their connectives are not shadow-lifting, and in general have ill-behaved gradients that do not guide the training as needed.

% \paragraph{Solution: \ac{QLL}}
% (1) ansatz good semantics via a naturality argument. observe how logits combine: conjunction (for independent variables) is addition while disjunction is given by log-sum-exp
% (2) the distinction independent/not, as well as the algebraic properties, are remindful of LL---indeed, fuzzy logics and LL are strictly related, as it is known by now \cite{metcalfeProofTheoryFuzzy2009,galatosResiduatedLatticesAlgebraic2007,cintulaHandbookMathematicalFuzzy2011a,}.
% thus LL is a good logic template to understand this semantics,
% however, LL and fuzzy logic still requires the additive connectives to coincide with \min and \max of numbers.
% to get around this problem, (3) notice there is in fact a family of connectives that interpolates between LSE and min, while preserving the algebraic properties (aside from idempotency) in particular their relation with addition.
% thus realize LSE is the additive connective of a \emph{quantitative LL}.

\paragraph{Problem: reconciling logic and machine learning needs.}
\Cref{tab:dl-feature-comparison} shows how the different aforementioned options stack against each other, as measured by both logical and analytic criteria.
Observe these two sets of criteria tend to exclude each other.

In particular, fuzzy logics are not designed to meet other analytic properties that are practically desirable for neuro-symbolic learning tasks, the most evident one being that connectives must be \emph{shadow-lifting}~\citep{varnaiRobustnessMetricsLearning2020}, that is, differentiable and entry-wise strictly increasing, so that their gradients can correctly inform the training dynamics.
This requirement turns out to be in direct tension with the algebraic picture above.
The gradient (where it exists) of the standard lattice operations $\min$ and $\max$, which interpret $\land$ and $\lor$, depends only on either one of the two operands, so that during training the other operand receives no signal at all.
\citet{varnaiRobustnessMetricsLearning2020} distilled this difficulty in a no-go theorem: \emph{no binary operation on the reals can be simultaneously idempotent, associative, and shadow-lifting}.

\paragraph{Solution: \ac{QLL}.}
We propose a resolution to this state of affairs in a logic we call \acfi{QLL}.
It combines many lessons learned from other differentiable logics, and organizes them according to a principle we call \emph{naturality}.
That means, its semantics uses operations which are already established in statistics, probability theory, or machine learning.

% The central observation is that the log-probability of the conjunction of two independent events is the sum of their log-probabilities.
% Thus, if log-probability is to be taken as a measure of belief or evidence, \textbf{addition gains a logical valence which differentiable logics must take into account}. Such `natural' operations are out-of-the-box compatible with the quantitative methods which constitute the rest of the training pipeline.

Naturality stems from the observation that even dimensionless numbers are implicitly \emph{typed} by the way they are used, and our choice of semantics must be type-correct in this sense.
% We will see in \cref{sec:evaluation} how obedience to this principle separates empirically effective methods from ineffective ones.
  %
% Intuitively, probabilities, likelihoods, perplexity, diversity indices are \emph{multiplicative} quantities, and are defined over $\MulReals = (\PosReals, \leq, 1, \Mtensor)$.
Intuitively, logits, information, entropy, energy, and distances are \emph{additive} quantities, as independent such quantities combine by sum.
Most commonly \emph{prediction loss functions} fall in this category (usually because they are either distances or entropies), thus our semantics must land in $\AddReals = (\Reals, \geq, 0, +)$ (note the reversed order).

We then proceed guided by \ac{LL}, a refinement of classical propositional logic introduced in the late 80s by~\citet{GIRARD19871,girardLinearLogicIts1995}\footnote{Though it draws from decades of tradition in \emph{relevance logic}, see \citep{avronSemanticsProofTheory1988}.}, in which connectives are distinguished between \emph{linear} (which aggregate two `independent' propositions) and \emph{non-linear} (which aggregate two propositions sharing the same `context')\footnote{Usually \emph{additive} and \emph{multiplicative}, a terminology we avoid for fear of confusion.}.
It is well-known that fuzzy logics can be seen as extension of \ac{LL}, see~\citep{cintulaHandbookMathematicalFuzzy2011,metcalfeProofTheoryFuzzy2009}.
Indeed, any logic whose semantics is numerical must necessarily treat truth values as \emph{resources} which cannot be freely duplicated or dropped: this is precisely what \ac{LL} does.
%
% \ac{LL} is now a standard logical framework in computer science, sporting an ever-growing range of applications from resource-sensitive type systems~\cite{bernardyLinearHaskellPractical2018} to communication protocols~\cite{cairesSessionTypesIntuitionistic2010} to differential privacy~\cite{gaboardiLinearDependentTypes2013}.
We can thus evince that \emph{addition} is the `linear conjunction' for additive quantities.
% These considerations are enough to single out the algebraic structure on $\Reals$ we want to target with the semantics, and we are left with the problem of making logical sense of these `natural operations'.
%
% We resolve this tension by reading off the right semantics from ML practice itself.
% This property is what we call \emph{naturality}.
% First, observe the log-probability of a conjunction of \emph{independent} events is the sum of their log-probabilities---so for independent quantities, addition behaves as a conjunction.
% Similarly, when one combines logits---the natural currency of neural classifiers---the corresponding disjunction is the following log-sum-exp function:
% \begin{equation}
%     a \pAor[1] b := -\log(e^{-a}+e^{-b}),
% \end{equation}
% Crucially, $+$ distributes over $\pAor$, a property that makes them stand in the same relation as the multiplicative and additive connectives of \ac{LL}.
% Together these single out an algebraic structure on $\AddReals = (\Reals,0,+)$, the home of loss functions and other additive quantities such as entropies, energies, and distances.%; multiplicative quantities (probabilities, likelihoods, perplexity) live dually in $\MulReals = (\PosReals,1,\Mtensor)$, related by Napierian duality.
Indeed, $(\Reals, \geq, 0, +, -, \max, \min)$ is a model of classical linear logic, but $\max$ and $\min$ are still the non-differentiable connectives we set out to avoid.

The final step is to \emph{soften} these operations, and consequently their algebraic skeleton.
We observe that there are in fact two one-parameter families of connectives---the \emph{`harmonic'} and plain \emph{log-sum-exp} (LSE) families $\pAor$ and $\pAand$, defined below---that converge to $\min/\max$ as $p\to\infty$ \citep{mitrinovicAnalyticInequalities1970} and are everywhere shadow-lifting for finite~$p$~\citep{nielsen_guaranteed_2016}.
\begin{equation}\label{eq:lses}
  a \pAor b = -\tfrac{1}{p}\log\!\bigl(e^{-pa}+e^{-pb}\bigr),
  \qquad
  a \pAand b = \tfrac{1}{p}\log\!\bigl(e^{pa}+e^{pb}\bigr)
\end{equation}
% This idea comes from \cite{varnaiRobustnessMetricsLearning2020}.
These connectives preserve every logically relevant algebraic property of $\min$ and $\max$ except idempotency; in particular, addition still distributes over them, which is exactly the property the \ac{LL} template requires.
Note, moreover, these are still natural as $\pAor[1]$ is the operation computing the logit of the disjunction of two events.
% Paying the price of \citeauthor{varnaiRobustnessMetricsLearning2020} impossibility by giving up the dispensable axiom (idempotency), we obtain shadow-lifting connectives with a quantitative convergence guarantee: at finite~$p$ training enjoys well-behaved gradients, while as $p\to\infty$ the surrogate converges to the original logical specification.
% % Annealing~$p$ during training thus traces a continuous path from a smoothed loss to the verbatim spec.
% We call the resulting system \emph{Quantitative Linear Logic} (\ac{QLL}): a differentiable logic in which the algebraic, the analytic, and the proof-theoretic requirements of a differentiable logic are satisfied simultaneously.

% The resulting logic, \ac{QLL}, is a `quantitative' variant of \ac{LL} in which additives are `softened' to operations like LSE.

\paragraph{Headline results.}
On every constraint we test (see~\cref{tab:constraints}), models trained with \ac{QLL} are formally verified to satisfy their specification at substantially higher rates than those trained with any competing logic.
The gap is starkest on classification robustness,
%(tested with perturbation $\varepsilon=0.2$),
where \ac{QLL} reaches \qty{46.0+-23.2}{\percent} verified satisfaction on MNIST and \qty{23.2+-9.6}{\percent} on Fashion-MNIST, while \ac{STL}, fuzzy logics, and the baseline are all stuck at \qty{0}{\percent}, and \ac{DL2} reaches at most \qty{10.4+-21.6}{\percent}.
On the clothing/footwear superclass anti-misclassification constraint on Fashion-MNIST, \ac{QLL} verifies \qty{97.4+-2.7}{\percent} of test instances against \qty{31.0+-19.0}{\percent} for \ac{DL2} and no more than \qty{28.4+-16.9}{\percent} for \ac{STL}. %; on the Dice \textsf{Not-Both} constraint, \qty{82.9+-7.2}{\percent} against \qty{16.2+-18.7}{\percent} for the best fuzzy logic and \qty{13.5+-16.2}{\percent} for \ac{DL2}.
%\Cref{tab:main_results-long}
\Cref{sec:evaluation} reports the full picture.

\paragraph{Main Contributions.}
\begin{enumerate}
    \item We define the syntax and semantics of \ac{QLL} (\cref{sec:qll}), a differentiable logic whose connectives are (1) by design compatible with the additive losses that already pervade ML practice (\emph{naturality}) and (2) preserve most of the familiar algebraic properties of classical and linear propositional logic (\emph{logical adequacy}).
    \item We establish the analytic properties of the \ac{QLL} connectives---smooth monotonicity, non-vanishing gradients, and a controlled ${p\to\infty}$ convergence to the verbatim logical specification---thereby threading the needle of \citeauthor{varnaiRobustnessMetricsLearning2020}'s no-go theorem~\citep{varnaiRobustnessMetricsLearning2020} and reconciling, for the first time, the algebraic, analytic, and proof-theoretic desiderata of a differentiable logic (\cref{sec:qll,tab:dl-feature-comparison}).
    \item We empirically evaluate \ac{QLL} against \ac{DL2}, \ac{STL}, and four fuzzy logics on MNIST, Fashion-MNIST, and a custom `Dice' benchmark, covering adversarial robustness, class-hierarchy constraints, and physical-world constraints (\cref{sec:evaluation,tab:main_results-long}). We assess each model along a hierarchy of increasingly rigorous metrics: %random-sample satisfaction (\acs{CAcc})\knote{no longer in the main text?},
    two versions of \ac{PGD}-based adversarial satisfaction (\acs{CSec}), and exhaustive formal verification (\acs{CSat}) with the Marabou verifier~\citep{katzMarabouFrameworkVerification2019,wuMarabouVersatileFormal2024} via Vehicle~\citep{daggittVehicleBridgingEmbedding2025}.
\end{enumerate}

%%
%% OLD INTRO IN BIN/NEW-INTRO-OLD.TEX
%%

% \paragraph{Outline.}
% \Cref{sec:background} reviews differentiable logics and property-driven training.%, framing the algebraic and analytic desiderata that motivate our design.
% \Cref{sec:qll} introduces \ac{QLL}, develops the additive reals as our semantic domain, establishes the analytic properties of the LSE-based connectives, and presents the syntax and semantics of the logic.
% \Cref{sec:evaluation} empirically compares \ac{QLL} against existing differentiable logics on benchmarkss based on the MNIST/Fashion-MNIST and Dice datasets, both at training time and under formal verification.
% \Cref{sec:discussion} discusses the broader implications of our results and outlines limitations and future work.

% \matteo{appendices}

\section{Background: Differentiable Logics and Property-Driven Training}
\label{sec:background}

Throughout the paper, we assume supervised machine learning scenarios that feature a data set $D$ with input-output pairs $(\hat{\vec{x}}, \hat{\vec{y}})$ with $\hat{\vec{x}} \in \real^m$ and $\hat{\vec{y}} \in \real^n$, sampled from a distribution $\mathcal{D}$.
Given a model $f_{\theta}: \real^m \to \real^n$, optimal parameters $\theta$ are obtained by gradient descent on an objective that combines a \emph{prediction loss} $\mathcal{L}_{\mathrm{pred}}: \real^n \times \real^n \to \real$ with a \emph{logical loss} $\dl{\phi}{}$ that penalises violations of a specification $\phi$ in an $\varepsilon$-neighbourhood of each data point:
\begin{equation}\label{eq:combined_optimisation}
    \argmin\limits_{\theta}\expected\limits_{(\hat{\vec{x}},\hat{\vec{y}})\sim\mathcal{D}}\ \biggl[\mathcal{L}_{\mathrm{pred}}(\hat{\vec{y}}, f_{\theta}(\hat{\vec{x}})) + \lambda \sup\limits_{\vec{x} \in B_\varepsilon(\hat{\vec{x}})} \dl{\phi}{}(\hat{\vec{y}}, f_{\theta}(\vec{x}))\biggr],
\end{equation}
where $B_\varepsilon(\hat{\vec{x}}) = \{\vec{x}\in\real^m \mid \lVert\hat{\vec{x}}-\vec{x}\rVert_\infty\le\varepsilon\}$ % = \{\vec{x}\in\real^m \mid \forall 1 \leq i \leq m, |\hat{x}_i - x_i|\le\varepsilon\}
is an $\varepsilon$-cube around the given data point $\hat{\vec{x}}$, and $\vec{x}$ thus stands in for an adversarial example.
The prediction loss $\mathcal{L}_{\mathrm{pred}}$ is usually cross-entropy for classification or mean-squared error for regression, while the logical loss $\dl{\phi}{}$ is supplied by the semantics function $\sem{-}$ of a \emph{differentiable logic}, which maps each well-formed $\phi$ to a differentiable function; \cref{tab:differentiable-logics-a,tab:differentiable-logics-b} survey the differentiable logics considered in this paper.
The contribution of each side to training is controlled by a parameter $\lambda \in [0,1]$---in practice we choose $\lambda \propto \frac{\lVert\nabla_{\theta}\mathcal{L}_{\mathrm{pred}} \rVert}{\lVert\nabla_{\theta}\sem{\phi} \rVert}$ so as to balance the norm of the gradients, thus ensuring that both losses contribute proportionately.

The second term of \eqref{eq:combined_optimisation} makes the optimisation problem \emph{adversarial}, and therefore optimisation is carried out via \acfi{PGD}~\citep{madryDeepLearningModels2018}: the inner supremum samples worst-case examples that violate the constraint $\phi$, and the outer minimisation then adapts $\theta$ to reduce the loss on those examples.

That second term can be logically interpreted as a universally quantified property ``for all $\vec{x} \in B_\varepsilon(\vec{\hat{x}}),\ \phi(f_\theta(\vec{x}, \hat{\vec{y}}))$''.

\paragraph{Related Work.}

Techniques for property-driven learning range from incorporating constraints into the training objective to enforcing them directly in the model architecture.
%from hard-coding some of the logical structure in the architecture of the model to a softer option of translating them into additional loss functions that penalise violations of the specification.
Early approaches include \ac{SBR}~\citep{diligentiSemanticbasedRegularizationLearning2017a}, Logic Tensor Networks~\citep{serafiniLogicTensorNetworks2016}, and probabilistic frameworks such as DeepProbLog~\citep{manhaeveDeepProbLogNeuralProbabilistic2018} and Semantic Loss~\citep{xuSemanticLossFunction2018}.
More recent efforts have focused on integrating logical constraints into a dedicated neural network layer, ensuring constraint satisfaction by design, e.g., CCN+~\citep{giunchigliaCCNNeurosymbolicFramework2024} and \ac{SPL}~\citep{ahmedSemanticProbabilisticLayers2022}, while approaches such as \ac{DRL}~\citep{stoianConvexityAssumptionRealistic2024} and \ac{PAL}~\citep{kurscheidtProbabilisticNeurosymbolicLayer2025} extend this idea to non-convex constraints. Our work addresses the methods that develop special purpose loss functions, and as such, our loss functions can complement any architectures, including the neuro-symbolic architectures mentioned above.

 \begin{wrapfigure}[18]{R}{0.5\linewidth}
     \centering
     \begin{tikzpicture}[
  every node/.style={font=\small}, %\sffamily
  arrow/.style={
    ->,
    >=Latex,
    shorten <=2pt,
    shorten >=2pt,
  },
  block/.style={text width=1.6cm, minimum height=0.8cm, align=center, draw, fill=white, drop shadow},
  invisibleBlock/.style={text width=1.8cm, align=center, font=\small\bfseries},
  arrowLabel/.style={font=\small\itshape, color=red},
]
  % Nodes
  \node [invisibleBlock] (neuralNetwork) {Neural\\Network $f$};
  \node [block, left=0.3cm of neuralNetwork] (training) {Training};
  \node [block, right=0.3cm of neuralNetwork] (verification) {Verification};
  \node [invisibleBlock, text width=4cm, above=1.7cm of neuralNetwork] (property) {Logical Specification $\phi$};

  % Arrows
  \draw[arrow, densely dashed] (training.north east) to [bend left] node[arrowLabel, midway, above]{verify} (verification.north west);
  \draw[arrow, densely dashed] (verification.south west) to [bend left] node[arrowLabel, midway, below]{retrain} (training.south east);

  \draw[arrow, shorten <=0pt] (property.south) to node[font=\small, midway, above, xshift=-0.4em]{$\sem{\phi}$} (training.north);
  \draw[arrow, shorten <=0pt] (property.south) to node[font=\small, midway, above, xshift=0.4em]{$\sem[\Bool]{\phi}$} (verification.north);
\end{tikzpicture}
     \caption{A diagrammatic depiction of the \emph{continuous verification} cycle~\citep{casadioNeuralNetworkRobustness2022,daggittVehicleBridgingEmbedding2025} commonly deployed in neuro-symbolic training and verification.}
     \label{fig:continuous-verification-cycle}
 \end{wrapfigure}
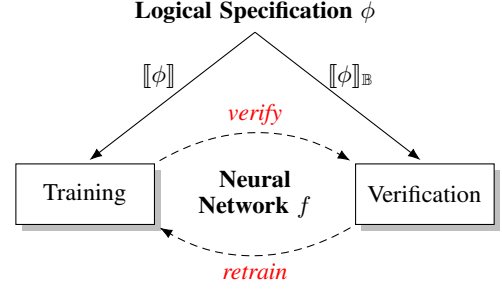
% %\end{figure}
On the side of formal verification of neural networks, a wide range of tools is available. %Approaches has been proposed, including
%in~\citep{huangSurveySafetyTrustworthiness2020,urbanReviewFormalMethods2021,liuAlgorithmsVerifyingDeep2021,albarghouthiIntroductionNeuralNetwork2021}.
They mainly differ by the algorithms they deploy: linear programming and \ac{SMT}-solving underlie Marabou~\citep{katzMarabouFrameworkVerification2019,wuMarabouVersatileFormal2024}, reachability-based methods underlie NNV~\citep{tranNNVNeuralNetwork2020,lopezNNV20Neural2023} and PyRAT~\citep{lemesleNeuralNetworkVerification2025,lemeslePyRATVerifyingNeural2026}, and linear bound propagation techniques gave rise to $\alpha,\beta$-CROWN~\citep{zhangEfficientNeuralNetwork2018a,xuAutomaticPerturbationAnalysis2020a,xuFastCompleteEnabling2021,wangBetaCROWNEfficientBound2021} (see also~\citep{huangSurveySafetyTrustworthiness2020,urbanReviewFormalMethods2021,liuAlgorithmsVerifyingDeep2021,albarghouthiIntroductionNeuralNetwork2021}). %\matteo{are these citations all necessary? this bit is very long}
%These methods typically assume pre-trained methods, for instance in benchmarks such as the annual \ac{VNNCOMP}~\citep{bakSecondInternationalVerification2021,mullerThirdInternationalVerification2022,brixFourthInternationalVerification2023,brixFirstThreeYears2023b,brixFifthInternationalVerification2024,kaulen6thInternationalVerification2025}.
In this work, we rely on Marabou, as it admits more general class of properties that some other solvers.
We use Marabou through Vehicle~\citep{daggittVehicleBridgingEmbedding2025}, a domain-specific language that provides a user-friendly interface for writing complex logical properties.%; and was built to enable the continuous verification cycle.

Given a logical specification for an ML model, one can either formally verify that it holds, or, failing that, retrain the model to satisfy it; and then make a verification attempt again.
This set up of using a single specification for both formal verification and training is commonly referred to as \emph{continuous verification}~\citep{daggittVehicleBridgingEmbedding2025,casadioNeuralNetworkRobustness2022} and is depicted in~\cref{fig:continuous-verification-cycle}.
The evaluation we present in~\cref{sec:evaluation} follows this pattern: training and verification are done with the same formula $\phi$.

\section{Quantitative Linear Logic}\label{sec:qll}

% In this section we describe the additive semantics of \ac{QLL}, a differentiable logic whose connectives arise from the natural algebraic structure of the real numbers~\citep{capucciQuantifiersQuantitativeReasoning2025}.\knote{As this maybe the first time you actually publish it, consider removing the arxiv citation, and just give it as original. You can restore the citation after acceptance}
% We first present the operations on the extended reals $\ExtReals$ and discuss their algebraic and analytical properties.
% We then define the propositional language used to write specifications, give its semantics, and explain how training loss relates to specification satisfaction.

% \subsection{The additive reals}\label{sec:add-reals}

We work in the extended reals $\ExtReals$, equipped with the \emph{reverse} order: $-\infty$ is maximally true and $+\infty$ is maximally false.
This convention matches the loss-function intuition that \textbf{lower numbers are more true}.
Indeed, losses are usually subject to minimization (think energy, entropy, distance).

As noted in the introduction, $\ExtReals$ is the natural domain of addition.
This will be our `primary' connective in this setting: all the other ones are obtained from this one following the template of Linear Logic\footnote{In a forthcoming work, we prove this is the semantics of a `quantitative' version of linear logic in which sequents are valued in real numbers.}~\citep{girardLinearLogicIts1995}.
The connectives so obtained are laid out in~\cref{table:add-reals}.

%  and are classified by \emph{polarity} (\emph{positive} or \emph{negative}) and \emph{linearity} (\emph{linear} or \emph{non-linear})\footnote{This terminology is non-standard, and these are usually called \emph{multiplicative} and \emph{additive}, respectively.}.
% The involution $a^* := -a$ exchanges polarity, so that one polarity determines the other by De Morgan duality, while linear connectives determine the non-linear ones by the requirement the former distribute over the latter.\footnote{Indeed, it can be shown there is essentially only one family of such operations \cite{aubrunMultiplicativePropertyCharacterizes2011}.}
% Therefore everything is determined by the choice of the \textbf{positive linear} connective.\knote{the last two sentences coudl be cut: will they know De Morgan, and you can just show the construction}

% We chose that to be ordinary addition $+$ extended to $\Reals$ by setting $-\infty + \infty = +\infty$ (so as to behave like \emph{conjunction} on $\{0,+\infty\}$).
% Then its De Morgan dual, the \textbf{negative linear} connective, is $a +^* b := -((-a) + (-b))$, which is the same as $+$ on all finite values except that $-\infty +^* \infty = -\infty$. (making it a \emph{disjunctive} operation).

% We will not use these operations directly here, but they are crucial since they determined \emph{implication} $a \mulimp b := b - a$, % by the residuation property $a + b \geq c \iff a \geq b \mulimp c$.
% which is key to encode inequalities (\cref{sec:qll-syntax}) but also to encode conditionals, see \cref{sec:evaluation}.

The \textbf{non-linear connectives} form two one-parameter families indexed by a \textbf{hardness} $p \in (0, \infty]$, as seen in \eqref{eq:lses}.
% For finite $p$, the positive non-linear connective is the \emph{log-sum-exp}~\citep{nielsen_guaranteed_2016} below right
% while its De Morgan dual operation is below left:
% \begin{equation}
% 	a \pAor b := -\tfrac{1}{p}\log(e^{-pa}+e^{-pb}),
% 	\qquad
% 	a \pAand b := \tfrac{1}{p}\log(e^{pa}+e^{pb}).
% \end{equation}
The first, $\pAor$ is \emph{disjunctive} while the second is \emph{conjunctive} $\pAand$.
The major feature of these operations is their smoothness, but it is equally crucial that they converge to the lattice connectives as $p \to +\infty$, thus for $p<\infty$ we can consider them as stand-ins for the lattice ones.

% \begin{proposition}\label{prop:approx}
% 	As $p \to \infty$, the non-linear connectives converge pointwise to the lattice operations:
% 	\(
% 		a \pAor b \conv[p \to \infty] \min(a, b),
%         \ %
%         a \pAand b \conv[p \to \infty] \max(a, b).
% 	\)
% \end{proposition}

% We thus define  $\pAor[\infty] = \min$ and $\pAand[\infty] = \max$ (recall the order is reversed).

\paragraph{Algebraic properties.}
We now substantiate the `good algebraic properties' attributed to \ac{QLL} in~\cref{tab:dl-feature-comparison}.
We start by noting both linear and non-linear operations (for all $p$) are strictly monotone, associative, and commutative, and have neutral (see~\cref{table:add-reals}) and absorbing elements.
Linear operations distribute over non-linear ones, but conjunctive operations do not necessarily distribute over disjunctive ones.
%\footnote{Though the inequality $a \pAand (b \pAor c) \geq (a \pAand b) \pAor (a \pAand c)$ holds, so that e.g. one can still convert a formula to disjunctive normal form (DNF).}
This distinguishes our operations from the differentiable robustness metric of \ac{STL} introduced in~\citep{varnaiRobustnessMetricsLearning2020}, whose connectives are not associative, which is a logically awkward property.

\begin{definition}
	The operations described in \Cref{table:add-reals} satisfy the properties of \cref{table:log-props}, plus unitality, associativity, and commutativity of all connectives (except $\mulimp$).
\end{definition}

We think of these as a `quantitative' version\footnote{Specifically, this is a poset enriched in $\ExtReals$ (cf.~\citep{lawvereMetricSpacesGeneralized1973}), plus extra algebraic structure making it a \emph{softale} in the sense of~\citep{capucciQuantitativeLinearLogic2026}.} of a \emph{residuated lattice} (see~\cref{subsec:residuated-lattices}); which is the algebraic structure for semantics of substructural logics (like fuzzy and linear).
The reader should compare it with the structure of Booleans $\Bool$, where also implication and order are given by the same operation $\Rightarrow$.

By the no-go theorem of~\citet{varnaiRobustnessMetricsLearning2020}, if we keep associativity and smooth monotonicity, we must drop some other property, and indeed for $p<\infty$ our connectives are not idempotent.
However, they are `idempotent up to a constant':
\begin{equation}
	\tfrac 1 p \log 2 \geq (a \pAor[p] a \mulimp a),
	\qquad
	\tfrac 1 p \log 2 \geq (a \mulimp a \pAand[p] a).
\end{equation}
We therefore conclude our operations are `logically adequate', since they respect the fundamental axioms of common algebraic semantics of substructural logic.

\begin{table}[tbp]
	\centering
	\footnotesize
	\caption{Algebraic structure of the additive reals $\AddReals$.
	Operations restricted to $\{-\infty,+\infty\}$ behave like either conjunction (blue background) or disjunction (red).
	% The order is reversed: $-\infty$ is maximally true and $+\infty$ is maximally false.
	Positive and negative polarities relate to each other via De Morgan duality.}
	\label{table:add-reals}
	\begin{tabularx}{\textwidth}{|c|c|*{3}{>{\centering\arraybackslash}X|}}
		\cline{2-5}
		\multicolumn{1}{c|}{} & \textbf{polarity} & \multicolumn{2}{c|}{\textbf{non-linear}} & \textbf{linear}\\
		\hline
		\multirow{2}{9ex}[-2ex]{\centering\parbox[c]{9ex}{\centering\textbf{negation}\\$a^* := -a$}} & \textbf{positive} &
		\multicolumn{2}{>{\columncolor{red!15}}c|}{\begin{tikzcd}[ampersand replacement=\&, column sep=small]
			\begin{array}{c}
				\begin{gathered}
					\bot := +\infty \\[-.5ex]
					a \pAor[\infty] b := \min(a, b)
				\end{gathered}
			\end{array}
			\&\&
			\begin{array}{c}
				\begin{gathered}
					\bot := +\infty \\[-.5ex]
					a \pAor b
				\end{gathered}
			\end{array}
			\arrow["{\ p \to \infty}"', from=1-3, to=1-1]
		\end{tikzcd}} & \cellcolor{blue!15}$\begin{gathered} \One := 0 \\[-.5ex] a + b \\[-.5ex] {\scriptstyle -\infty + \infty := +\infty}\end{gathered}$\\[2ex]
		\hhline{~|----|}
		& \textbf{negative} & \multicolumn{2}{>{\columncolor{blue!15}}c|}{\begin{tikzcd}[ampersand replacement=\&, column sep=small]
			\begin{array}{c}
				\begin{gathered}
					\top := -\infty \\[-.5ex]
					a \pAand[\infty] b := \max(a, b)
				\end{gathered}
			\end{array}
			\&\&
			\begin{array}{c}
				\begin{gathered}
					\top := -\infty \\[-.5ex]
					 a \pAand b
				\end{gathered}
			 \end{array}
			\arrow["{\ p \to \infty}"', from=1-3, to=1-1]
		\end{tikzcd}} & \cellcolor{red!15}$\begin{gathered} \One := 0 \\[-.5ex] a +^* b \\[-.5ex] {\scriptstyle -\infty +^* \infty := -\infty}\end{gathered}$\\[2ex]
		\hline
	\end{tabularx}

	\vspace*{2ex}
	\textbf{order}: $a \geq b$,
	\hspace*{6ex}
	\textbf{implication}: $a \mulimp b = a^* +^* b$, thus $b - a$.
	%, \qquad \textbf{scaling}: $k \in [0,\infty],\quad k \ast a = k \Mtensor a$
\end{table}

\begin{table}[tbp]
	\centering
	\caption{Logical properties satisfied by $\AddReals$ (\cref{table:add-reals}), viewed as a quantitative residuated lattice (a \emph{softale}) with a real-valued order ${\mulimp}:\ExtReals^2 \to \ExtReals$ coinciding with implication.}
	\label{table:log-props}
	\renewcommand{\arraystretch}{1.3}%
	\resizebox{\textwidth}{!}{%
	\normalfont\footnotesize
		\begin{tabular}{|c|c|c|}
			\hline
			\begin{tabular}{@{}c@{}}\textbf{identity}\\ $0 \geq (a \mulimp a)$\end{tabular}
			& \begin{tabular}{@{}c@{}}\textbf{modus ponens}\\ $(a \mulimp b) + (b \mulimp c) \geq (a \mulimp c)$\end{tabular}
			& \begin{tabular}{@{}c@{}}\textbf{contraposition}\\ $a \mulimp b = b^* \mulimp ^*$\end{tabular} \\
			\hline
			\begin{tabular}{@{}c@{}}\textbf{residuation} (currying)\\ $(a + b \mulimp c) = (a \mulimp (b \mulimp c))$\end{tabular}
			& \begin{tabular}{@{}c@{}}\textbf{mix}\\ $(a \mulimp b) + (c \mulimp d) \geq (a+c) \mulimp (b + d)$\end{tabular}
			& \begin{tabular}{@{}c@{}}\textbf{prelinearity}\\ $(a \mulimp b)^* \leq (b \mulimp a)$\end{tabular} \\
			\hline
			\begin{tabular}{@{}c@{}}\textbf{isomix}\\ $0^* = 0$\end{tabular}
			& \begin{tabular}{@{}c@{}}\textbf{involutivity}\\ $a^{**} = a$\end{tabular}
			& \begin{tabular}{@{}c@{}}\textbf{bottom}\\ $+\infty \mulimp a = +\infty$\end{tabular} \\
			\hline
			\multicolumn{3}{|c|}{\begin{tabular}{@{}c@{}}
				\textbf{disjunction rules}\footnote{One gets analogous ones for conjunction by dualization.} \\
				$(a \mulimp c) \pAand (b \mulimp c) \geq (a \pAor b \mulimp c)$
				\quad and \quad
				$(c \mulimp a) \pAor (c \mulimp b) \geq (c \mulimp a \pAor b)$
			\end{tabular}} \\
			\hline
		\end{tabular}%
	}
\end{table}

\paragraph{Analytic properties.}%\label{sec:analytic-props}
We now substantiate the `analytic properties' attributed to \ac{QLL} in~\cref{tab:dl-feature-comparison}.
Indeed, the non-linear connectives of QLL have remarkably well-behaved and interpretable partial derivatives:%\footnote{Note that these read very naturally as expressions in the multiplicative reals: indeed, reals with addition form a Lie group whose Lie algebra (tangent at $0$) are reals with multiplication.}:%, as we show in \cref{fig:gradient-fields}:
\begin{equation}\label{eq:derivatives-lse}
	\pdv{x \pAor y}{x} = \frac{e^{py}}{e^{px}+e^{py}},
	\hspace*{10ex}
	\pdv{x \pAand y}{x} = \frac{e^{px}}{e^{px}+e^{py}}.
\end{equation}
We plot the functions and their gradients in \cref{tab:conjunction-disjunction,tab:QLL_STL_conjunction,tab:QLL_STL_grad_vs_max}.

Several features conspire to make these gradients well-suited for optimization.
First, the domain $\ExtReals$ is unbounded: unlike connectives valued in $[0,1]$ or $[0,+\infty)$, no amount of satisfaction or violation drives operands into a saturated region where derivatives flatten out.
Second, the gradients align with logical intuition, i.e. conjunction gets gradients which focus the optimisation effort on the most violated term.
Third, the partial derivatives never vanish at finite operands---the contribution of every subformula remains visible to the optimizer regardless of how thoroughly it is dominated by its peers.
This is the \emph{shadow-lifting} property of \citet{varnaiRobustnessMetricsLearning2020}, and it is what justifies preferring $\pAor/\pAand$ to their non-smooth counterparts $\min/\max$:

\begin{proposition}[Shadow-lifting]\label{prop:shadow-lifting}
	For $p< \infty$, the partial derivatives of $\pAand$ and $\pAor$ are strictly positive for all finite values---no operand is ever `shadowed' by the other.
\end{proposition}

Together, these properties give rise to what we call \emph{adversarially informative gradients}.
In adversarial min-max optimisation, gradients must decrease as constraints become satisfied, yet not vanish entirely so as to still expose nearby violations.
% The unbounded domain of \ac{QLL} connectives avoids saturation, while shadow-lifting ensures that no operand is ignored.
As a consequence, gradient signal persists even for satisfied constraints, enabling gradient-based methods such as \ac{PGD} to find violations while still guiding optimisation towards satisfaction.

\subsection{Syntax and semantics}\label{sec:qll-syntax}
\begin{comment}
Consider the following operation, called \textbf{Booleanization}:
\begin{equation}
	a \in \ExtReals, \qquad
	\booleize a = \begin{dcases}
		0 & a \leq 0\\
		+\infty & \text{otherwise}
	\end{dcases}
\end{equation}
This operation can be considered a left inverse to the inclusion $\Bool \to \ExtReals$ sending $\true$ to $0$ and $\false$ to $\infty$, and indeed we often denote such a map again by $\booleize(-)$.
% In the multiplicative reals, the same operation is defined as $\booleize a = 1$ when $a \geq 1$ and $0$ otherwise, which agrees with how indicator functions for sets are defined, in the sense that for a function $\phi:X \to \PosReals$, $\booleize \phi(x) = \phi(x)$ if and only if $\phi$ is the indicator function of some subset $A \subseteq X$, hence the name.

\begin{definition}[Boolean predicate]\label{def:bool-pred}
	We will say that a function $\phi:X \to \PosReals$ is \textbf{Boolean} if $\booleize \phi(x) = \phi(x)$, in which case it corresponds to some subset $A = \{x \in X \mid \phi(x) \leq 0\}$.
\end{definition}

We abuse notation and denote by $\booleize A$ the corresponding Boolean predicate in the above sense.

\begin{remark}
	The analogous of a Boolean function for the multiplicative reals is much more familiar: those are precisely the characteristic functions of subsets.
\end{remark}
\end{comment}

We now define the propositional fragment of \ac{QLL} used for constructing differentiable loss functions.
The syntax features propositional variables (or `atomic formulae') $y_i$ and $\hat{y}_i$ that stand for, respectively, the $i$-th output of the model and the $i$-th entry of the label of a data point.

\begin{definition}\label{def:qll-syntax}
	For our NeSy specifications, we consider the following subset of the \ac{QLL} language:
	\begin{equation}
		 \phi \grammareq y_i \grammarsep \hat{y}_i \grammarsep \Fnot{\phi} \grammarsep \phi \Fand \phi \grammarsep \phi \For \phi \grammarsep \phi \mulimp \phi
	\end{equation}
	where $1 \leq i \leq n$, and $r$ ranges over\footnote{Without loss of practicality one can restrict $r$ to rationals or a finite subset of machine-representable numbers.} $\ExtReals$.
\end{definition}

\begin{remark}
	This is a simplified version of the language in which atomic propositions of the form $r$ also stand in for the logical constants $\bot$ (which `means' $+\infty$), $\top$ (which `means' $-\infty$), and $\One$ (which `means' $0$).
	The full language can be found in \citep{capucciQuantifiersQuantitativeReasoning2025}.
\end{remark}

We will also often write $\bigwedge_{a \in A} \phi_a$, where $A$ is a finite set and $\{\phi_a\}_{a \in A}$ is a family of formulae, to abbreviate a long but finitary conjunction.

\begin{definition}[Additive semantics]\label{def:additive-semantics}
	The additive semantics $\sem[p]{\phi} \in \ExtReals^{\R^{2n}}$ is defined inductively on the structure of the formula, where $\pi_i$ denotes the $i$-th coordinate projection of $\R^{2n} \cong \R^{n + n}$:
	\begin{equation}\label{eq:qll-semantics}
		\begin{gathered}
			\sem[p]{\hat{y}_i} = \pi_{i},
			\qquad
			\sem[p]{y_i} = \pi_{n+i},
			\qquad
			\sem[p]{r} = r,
			\qquad
			\sem[p]{\Fnot{\phi}} = -\sem[p]{\phi},
			\\[.5ex]
			\sem[p]{\phi \Fand \psi} = \sem[p]{\phi} \pAand \sem[p]{\psi},
			\quad
			\sem[p]{\phi \For \psi} = \sem[p]{\phi} \pAor \sem[p]{\psi},
			\quad
			\sem[p]{\phi \mulimp \psi} = \sem[p]{\psi} - \sem[p]{\phi}
		\end{gathered}
	\end{equation}
\end{definition}

Since $\mulimp$ is both a residuation for the lattice $(\ExtReals, \geq)$ and the `quantitative' order structure, it encodes the order.
Therefore we can use it in formulae to account for $\leq/\geq$.
We think of this as a feature, but also note one could add an extra connective $\phi \leq \phi$ to be interpreted as necessary (e.g. ReLU, softplus, etc.).

% \begin{example}[Robustness]\label{ex:robustness}
% 	The logical part of \cref{eq:robust} can be translated to \ac{QLL} as
%     \begin{equation}
%         \rho := \bigwedge_{1 \leq i \leq n} ((\hat{y}_i \mulimp y_i) \mulimp \delta) \land ((y_i \mulimp \hat{y}_i) \mulimp \delta)
%     \end{equation}
%     We explain in which sense this is the correct translation later in \cref{ex:robustness-explained}
% \end{example}

Many examples of \ac{QLL} formalization can be found in \cref{subsec:qll-property-formalisations}.

In order to use \ac{QLL} to write specifications, it is necessary to be able to relate the semantics of the \ac{QLL} formula (which is quantitative) to its satisfaction in the Boolean sense.

\begin{definition}[Boolean semantics]\label{def:boolean-semantics}
	The boolean semantics $\sem[\Bool]{\phi} \in \Bool^{\R^{2n}}$ is defined inductively on the structure of the formula (inequalities are intended pointwise):
	\begin{equation}\label{eq:boolean-semantics}
		\begin{gathered}
			\sem[\Bool]{\hat{y}_i} = (\pi_i \leq 0),
			\qquad
			\sem[\Bool]{y_i} = (\pi_{n+i} \leq 0),
			\qquad
			\sem[\Bool]{r} = (r \leq 0),
			\qquad
			\sem[\Bool]{\Fnot{\phi}} = (\sem[\infty]{\phi} \geq 0),
			\\[.5ex]
			\sem[\Bool]{\phi \Fand \psi} = \sem[\Bool]{\phi} \land \sem[\Bool]{\psi},
			\quad
			\sem[\Bool]{\phi \For \psi} = \sem[\Bool]{\phi} \lor \sem[\Bool]{\psi},
			\quad
			\sem[\Bool]{\phi \mulimp \psi} = \left(\sem[\infty]{\psi} \leq \sem[\infty]{\phi}\right).
		\end{gathered}
	\end{equation}
\end{definition}

The following is then trivial to prove:

\begin{theorem}[Soundness]\label{thm:soundness}
	Let $\phi$ be a formula in the language defined in \cref{def:qll-syntax}.
	Then
	\begin{equation}\label{eq:soundness-claim}
		\sem[\Bool]{\phi} = \true
		\quad\Longleftrightarrow\quad
		\sem[\infty]{\phi} \leq 0.
	\end{equation}
\end{theorem}
\begin{proof}
	By induction on the structure of $\phi$, see \Cref{app:soundness}.
\end{proof}

% Indeed, the above clauses are chosen specifically. for $\Fand,\For,\Fnot$ follow the standard translation of classical linear logic formulae into classical propositional ones.
% The clause for $\mulimp$ is instead quantitative: in $\AddReals$, implication encodes the order via residuation ($a \mulimp b \leq 0 \iff b \leq a$), and we transport this directly to the Boolean side rather than collapsing it to material implication.
% This is harmless on Boolean-valued atoms (where the two coincide) but strictly stronger on real-valued ones, and it is what makes \cref{ex:robustness-explained} go through.
% With this in place we obtain:

Compare to \ac{DL2}, for which \citep[Theorem~1]{fischerDL2TrainingQuerying2019} says $\phi$ is satisfied iff its semantics is \emph{exactly} $0$.

\section{Evaluation of QLL}\label{sec:evaluation}
In this section we evaluate \acs{QLL} against other differentiable logics based on their ability to jointly achieve high prediction performance and satisfy a given logical constraint on a classification task.
As a baseline, we consider standard supervised learning without any differentiable logics or adversarial training, where models are trained solely to optimise prediction accuracy.
Details on the experimental setup, including datasets, model architectures, and evaluation methodology, are provided in~\cref{subsec:additional-experimental}.

The experiments were run using a reusable, general purpose framework for property-driven training we implemented in PyTorch~\citep{paszkePyTorchImperativeStyle2019}.
The framework, including all code, data, and scripts required to reproduce the results presented in this paper, is available at \url{https://github.com/tflinkow/property-driven-ml}.

\subsection{Experimental Set Up: Exhaustive Verification vs Testing and Adversarial Attacks}

Given a logical constraint $\phi$, there are three ways in which adherence of a model to the constraint can be measured.
We list them in the order of increasing rigor.
The standard mathematical definitions of \acfi{CSec} and  \acfi{CSat}, following~\citep{casadioNeuralNetworkRobustness2022}, are given in~\cref{def:cacc,def:csec,def:csat} in~\cref{subsec:evaluation-metrics}.
\begin{itemize}[leftmargin=*]

    \item \acfi{CSec} evaluates whether a \ac{PGD} attack~\citep{madryDeepLearningModels2018} can discover counterexamples to $\phi$ in $B_\varepsilon(\hat{\vec{x}})$. This is the \emph{de facto} standard ML method of finding examples violating $\phi$, see~\citep{fischerDL2TrainingQuerying2019}.
    Effectiveness of this search varies based on the loss, and it comes with no formal guarantee of finding a counterexample if one exists\footnote{We make use of the AutoPGD~\citep{croceReliableEvaluationAdversarial2020} attack, since the effectiveness of standard \ac{PGD} has been shown to strongly rely on good attack parameters~\citep{mosbachLogitPairingMethods2019,croceScalingRandomizedGradientFree2020}}.
    \item \acfi{CSat} evaluates whether $\phi$ holds for all possible samples in $B_\varepsilon(\hat{\vec{x}})$.
    This measurement, also called \emph{formal verification}, is the most rigorous, and gives sound and complete guarantees of finding a counterexample if one exists.
    Formal verification is carried out with Marabou~\citep{katzMarabouFrameworkVerification2019,wuMarabouVersatileFormal2024} and Vehicle~\citep{daggittVehicleBridgingEmbedding2025}.
\end{itemize}

All metrics are evaluated on the test sets.
Standard prediction accuracy of models is denoted \acfi{PAcc}.
The \ac{CSat} metric, which exhaustively analyses the entire $B_\varepsilon(\hat{\vec{x}})$ for a data point $\hat{\vec{x}}$, is obtained from a fixed randomly chosen subset of the test set.
For all experiments, we report the mean and standard deviation (in \unit{\percent}) across $5$ random seeds for all metrics.

\emph{Given a property $\phi$, \ac{CSec}, seen as robustness to attack on $\phi$, is thus always an optimistic approximation of \ac{CSat} for $\phi$.
Our experiments in \Cref{tab:main_results-short,tab:main_results-long} show that typically, this over-approximation is vastly unrepresentative.
However, the gap between \ac{CSec} and \ac{CSat} is consistently lower when we train with \ac{QLL}.}
%
%The self oracle measures how effectively a logic can find violations using its own semantics.
%The reference oracle is not intended as an objective ground truth, but only as a shared empirical comparison across different logics.

%For all experiments, we report the mean and standard deviation (in \unit{\percent}) across $5$ random seeds for \acfi{PAcc}, \acfi{CAcc} (constraint satisfaction on random samples), and \acfi{CSec} (constraint satisfaction on adversarial samples).

%Importantly, adversarial attacks measure not directly whether a property is true, but only whether a violation can be found by a particular optimisation method.
%Consequently, a high \ac{CSec} score may be either because the network indeed satisfies the constraint, or trivially because the attack fails to find violations.

%To obtain exhaustive guarantees that the model satisfies the constraint, we additionally utilise the formal verifier Marabou~\citep{katzMarabouFrameworkVerification2019,wuMarabouVersatileFormal2024} via Vehicle~\citep{daggittVehicleBridgingEmbedding2025}, and report \acfi{CSat}, the fraction of instances for which the model is proven to satisfy the constraint.
%All constraints are expressed directly on logits due to solver limitations~\citep{wuMarabouVersatileFormal2024}.
% TODO:: irrelevant?

%{\color{red}
%We highlight in bold-face the best-performing logic per table, by which we mean the logic that achieves the highest product of mean \ac{PAcc} and mean \ac{CSat} at the $\epsilon$ used during training.
%}

\subsection{Performance across different kinds of NeSy tasks}

\begin{table}
    \footnotesize
    \centering
    \caption{Names and definitions of four properties $\phi$ that were used in the continuous verification experiments in this paper. 
    %used in training and verification with ``$\text{for all $\vec{x} \in B_\varepsilon(\hat{\vec{x}})$}, \phi(f_\theta(\vec{x}), \hat{\vec{y}})$''. 
    Their \ac{QLL} formalisations are presented in~\cref{subsec:qll-property-formalisations}, except for the Dice one which is simple enough to not require further discussion.
    Recall the semantics of $y_i$ is $f_\theta(\vec{x})_i$, for all labels $i$, and where $\vec{x} \in B_\varepsilon(\hat{\vec{x}})$.}
    \label{tab:constraints}
    \begin{tblr}
    {
      width=\linewidth,
      colspec={Q[l, m, mode=text, wd=2.3cm]Q[l, m, mode=math]X[l, m, mode=text]},
      row{1}={mode=text, font=\bfseries},
    }
      \toprule
        Name & Property & Description\\
      \midrule
        Strong classification robustness & y_{c} \geq\delta & the true-class logit must be above the fixed threshold $\delta$\\
        Classification robustness & \text{for all $i \in L$}, y_c \geq y_i & the true-class logit must be the largest\\
        Clothing/Footwear & \begin{cases}
            \bigwedge_{i\in F} y_{c} \geq y_i,&c\in C,\\
            \bigwedge_{i\in C} y_{c} \geq y_i,&c\in F
        \end{cases} & if the ground-truth class belongs to the group of clothing, its logit must exceed all footwear logits, and vice versa\\
        %
        %$\mathsf{NotBoth}$ & \bigwedge_{(i,j)\in P}(\lnot y_i \lor \lnot y_j) & for each pair of opposite faces, both faces must not be predicted at the same time \\
        %
        Exactly-One & \bigwedge_{(i,j)\in P}
		(\lnot y_i \lor \lnot y_j)\land(y_i \lor y_j) & for each pair of opposite faces, predict exactly one \\
      \bottomrule
    \end{tblr}%
\end{table}

\begin{table}
  \scriptsize
  \begin{talltblr}
  [
    theme=paper,
	caption={Selected experimental results for training and verifying models on MNIST, Fashion-MNIST, and Dice using different constraints and differentiable logics. Boldface indicates the best-performing configuration in each experiment, measured by product of PAcc and CSat. When a model is trained with a choice of $\varepsilon$, it is then attacked (\acs{CSec}) and verified (\acs{CSat}) for the same $\varepsilon$.},
	label={tab:main_results-short},
    note{a}={The \enquote{unknown} column denotes verification instances that could neither be proved or disproved within a \qty{30}{\second} timeout.},
    note{b}={We tested baseline CSec against QLL, see \cref{subsec:evaluation-metrics} for an explanation.}
  ]
  {
    width=\linewidth,
	colspec={X[1.7, l, m]X[l, m]l*{4}{S[table-format=3.1(2.1)]}},
    rowsep=1pt,
    row{1}={guard, font=\bfseries, mode=text},
	row{2}={guard, mode=text},
    row{6,14,17,18,26}={font=\bfseries},
  }
	\toprule
    \SetCell[r=2,c=1]{l} Constraint
    &
    \SetCell[r=2,c=1]{l} Dataset
    &
	  \SetCell[r=2,c=1]{l} Logic
	  &
	  \SetCell[r=2,c=1]{c} {\acs{PAcc} (\unit{\percent})}
	  &
	  \SetCell[r=2,c=1]{c} {\acs{CSec}$_{\epsilon}$ (\unit{\percent})}
	  &
	  \SetCell[r=1,c=2]{c} \acs{CSat}$_{\epsilon}$ (\unit{\percent})
	  &
	  \\
	\cmidrule[lr]{6-7}
	  & & & & & verified & unknown\TblrNote{a}\\
	\midrule
      \SetCell[r=4,c=1]{l} {Clothing/Footwear\\(training with $\epsilon=0.1$)} & \SetCell[r=4,c=1]{l} {Fashion-MNIST} & Baseline & 87.4+-0.5 & 35.8+-11.1\TblrNote{b} & 19.0+-0.0 & 1.8+-2.5 \\
      & & DL2 & 86.9+-0.8 & 99.8+-0.0 & 31.0+-19.0 & 36.6+-11.8 \\
    \cmidrule[lr]{3-7}
      & & ${}_5$\acs{STL} & 87.2+-0.5 & 99.8+-0.0 & 25.8+-11.9 & 35.0+-5.5 \\
      & & ${}_5$\acs{QLL} & 85.2+-1.7 & 99.0+-0.3 & 87.2+-12.3 & 12.0+-11.7 \\
    \midrule
      \SetCell[r=8,c=1]{l} {Exactly-One\\(training with $\epsilon=\nicefrac{4}{255}$)} & \SetCell[r=8,c=1]{l} Dice & Baseline & 83.8+-4.7 & 25.6+-7.2\TblrNote{b} & 6.2+-6.2 & 5.6+-3.0 \\
      & & DL2 & 83.7+-2.1 & 90.0+-4.5 & 3.2+-3.0 & 6.5+-6.7 \\
    \cmidrule[lr]{3-7}
      & & Gödel & 83.4+-1.6 & 85.8+-4.5 & 2.9+-4.2 & 10.3+-2.1 \\
      & & \L ukasiewicz & 84.4+-3.2 & 92.2+-5.2 & 5.9+-6.2 & 8.1+-3.1 \\
      & & Product & 84.2+-1.9 & 92.2+-3.1 & 3.4+-3.4 & 10.3+-2.9 \\
      & & ${}_2$Yager & 85.9+-2.1 & 91.2+-6.7 & 3.4+-3.1 & 6.4+-0.8 \\
    \cmidrule[lr]{3-7}
      & & ${}_5$\acs{STL} & 82.9+-3.4 & 41.8+-12.7 & 10.6+-9.1 & 4.4+-2.1 \\
      & & ${}_5$\acs{QLL} & 82.5+-4.4 & 51.5+-14.4 & 14.7+-10.3 & 5.9+-5.4 \\
     \midrule
      \SetCell[r=4,c=1]{l} {Strong classification robustness\\(training with $\epsilon=0.1$)} &\SetCell[r=4,c=1]{l} MNIST & Baseline & 97.6+-0.4 & 19.3+-18.1\TblrNote{b} & 2.0+-1.9 & 14.2+-18.1 \\
      & & DL2 & 97.6+-0.4 & 100.0+-0.0 & 24.2+-10.3 & 29.2+-19.8 \\
      & & ${}_5$\acs{STL} & 97.5+-0.3 & 100.0+-0.0 & 93.4+-9.4 & 6.6+-9.4 \\
      & & ${}_5$\acs{QLL} & 97.5+-0.3 & 100.0+-0.0 & 93.4+-9.4 & 6.6+-9.4 \\
    \midrule
	   \SetCell[r=8,c=1]{l} {Classification robustness\\(training with $\epsilon=0.2$)} & \SetCell[r=8,c=1]{l} MNIST & Baseline & 97.6+-0.4 & 0.0+-0.0\TblrNote{b} & 0.0+-0.0 & 0.0+-0.0 \\
	   & & DL2 & 97.1+-1.2 & 70.4+-34.0 & 10.4+-21.6 & 26.4+-32.8 \\
    \cmidrule[lr]{3-7}
      & & Gödel & 97.5+-0.4 & 94.7+-1.6 & 0.0+-0.0 & 0.4+-0.5 \\
      & & \L ukasiewicz & 97.6+-0.3 & 92.3+-5.1 & 0.0+-0.0 & 0.2+-0.4 \\
      & & Product & 97.6+-0.4 & 94.3+-3.0 & 0.0+-0.0 & 0.5+-0.6 \\
      & & ${}_2$Yager & 97.5+-0.6 & 93.3+-2.7 & 0.0+-0.0 & 0.2+-0.5 \\
    \cmidrule[lr]{3-7}
	   & & ${}_5$\acs{STL} & 97.7+-0.5 & 68.5+-34.9 & 0.0+-0.0 & 4.0+-6.2 \\
	   & & ${}_5$\acs{QLL} & 96.1+-0.8 & 67.7+-4.0 & 46.0+-23.2 & 30.0+-28.8 \\
	\bottomrule
  \end{talltblr}
\end{table}

Our evaluation covers three representative categories of neuro-symbolic learning tasks introduced in~\cref{sec:introduction}.
We evaluate \emph{adversarial robustness} on MNIST~\citep{lecunGradientbasedLearningApplied1998} and Fashion-MNIST using different robustness constraints from the literature.
We encode \emph{semantic class relations} on Fashion-MNIST~\citep{xiaoFashionMNISTNovelImage2017}, requiring a clear separation between clothing items and footwear.
For \emph{physical world-inspired constraints}, we use an original data set, and encode constraints arising from the geometry of the dice, such as not predicting two opposite faces at the same time.
Descriptions of these properties are presented in~\cref{tab:constraints}, with their \ac{QLL} formalisations described in~\cref{subsec:qll-property-formalisations}, and the datasets are described in~\cref{subsec:models-datasets}.
%\knote{Somewhere here or in the next subsection, we need to explain how comparisons and implications relate. I do not think we ever say it clearly...}

%\tnote{actually, maybe this is the place where to move all constraint descriptions instead of 4.3, 4.4., 4.5?}

\emph{Across all experiments (\Cref{tab:main_results-short,tab:main_results-long}), QLL consistently performs better than any other differentiable logic, irrespective of the chosen type of NeSy set up}.

%To enable a fair comparison despite fundamentally different loss scales and optimisation behaviour, we adopt a unified training and evaluation approach across all experiments, described below.

\subsection{Effect of the Interval Choice}\label{subsec:int}
Prior differentiable logic design made different choices for the interval on which loss functions are defined: fuzzy differentiable logics~\citep{vankriekenAnalyzingDifferentiableFuzzy2022} work on the interval $[0,1]$, DL2~\citep{fischerDL2TrainingQuerying2019}---on $[0,\infty)$, and STL~\citep{varnaiRobustnessMetricsLearning2020}---on $(-\infty, \infty)$. Recall that in QLL, the choice of the interval $\ExtReals$ is motivated by naturality, specifically that it is the range for cross-entropy, which is used as prediction loss.
Thus, the choice of the interval and the choice of operations are---via the naturality principle---the two sides of one coin. To test the effect of the interval choice alone, it suffices to compare the performance of STL/QLL vs DL2 on the strong classification robustness property that does not involve any conjunctions or disjunctions (all three logics interpret comparisons in the same way).

\emph{As \Cref{tab:main_results-short,tab:main_results-long} show, STL and QLL consistently out-perform DL2 on SCR.} We hypothesise that two factors are at play---both the range of the interval that matches better the standard cross-entropy loss, but also avoidance of $0$ gradients. Note that in DL2 and fuzzy logics, as soon as a constraint is satisfied, gradients vanish.
%Then, satisfaction of the constraint on one $\varepsilon$-cube stops effective training for all other $\varepsilon$-cubes.
Indeed, note that the inner maximisation step in~\eqref{eq:combined_optimisation} requires informative gradients to discover nearby violations even when the constraint is currently satisfied for one data point.
QLL and STL avoid this problem, by having the entire sub-interval $[-\infty,0]$ standing for satisfaction (recall \cref{thm:soundness}).

\subsection{Effect of Convergence  and Softness Parameters}
The next property of \ac{QLL} that we hold valuable is the fact that non-linear connectives in both logics converge to $\min$ and $\max$ as $p \rightarrow \infty$.
\ac{STL} shares the same feature, whereby its connectives are parameterised by a parameter $\nu$ with the same limit as $\nu \to \infty$.
Our hypothesis is that this property conveys better the effect of classical connectives at training time.

\emph{\Cref{tab:main_results-short,tab:main_results-long} show that having additive approximations of $\min$/$\max$ is useful: on three out of four properties, STL outperforms most other logics and QLL outperforms all other logics.}
We note, however, that current experiments are inconclusive when we try to isolate the effect of the interval choice from the effect of $\min$ and $\max$ approximations, where the former might overshadow the latter. This can be seen on the fuzzy logics, where Gödel directly uses $\min$ and $\max$, Yager smoothly approximates $\min$ and $\max$, and product uses neither. Still, all three fuzzy logics perform similarly.

\subsection{Gradients and Effect of Shadow-Lifting}

\begin{table}
    \small
    \centering
    \footnotesize
    \caption{Comparison of $\nabla\sem{x \land y}$ against $\nabla\max$ (defined as $(1/2,1/2)$ on the diagonal) for the \acs{STL} ($\nu=2$) and \acs{QLL} ($p=2$) conjunctions. Top row shows discrepancy in magnitude (lower is better), bottom row discrepancy in direction (closer to $1$ is better).}
    \label{tab:QLL_STL_grad_vs_max}
    \begin{tblr}
    {
      width=\textwidth,
      colspec={Q[c, m]X[c, m]X[c, m]},
      row{1}={mode=text, font=\bfseries},
    }
      \toprule
        & \acs{STL} ($\nu=2$) & \acs{QLL} ($p=2$) \\
      \midrule
        \adjustbox{valign=m}{\rotatebox{90}{$\lVert\nabla\sem{x \land y}-\nabla\max\rVert$}}
          & \includegraphics[width=.8\linewidth, valign=m]{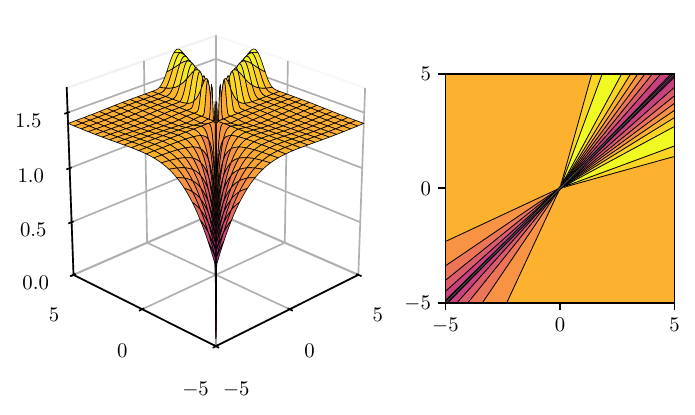}
          & \includegraphics[width=.8\linewidth, valign=m]{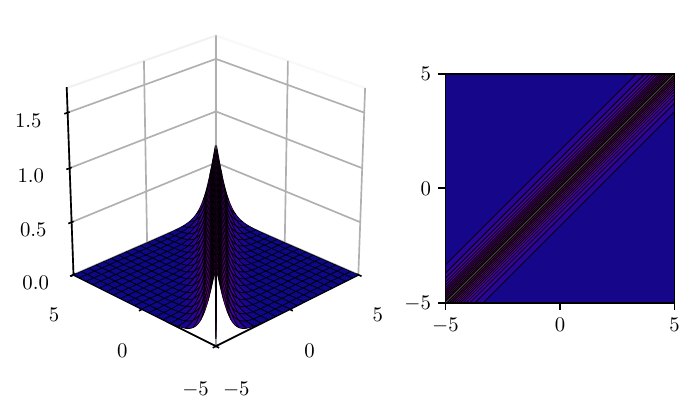}\\
        \adjustbox{valign=m}{\rotatebox{90}{$\cos(\nabla\sem{x \land y},\nabla\max)$}}
          & \includegraphics[width=.8\linewidth, valign=m]{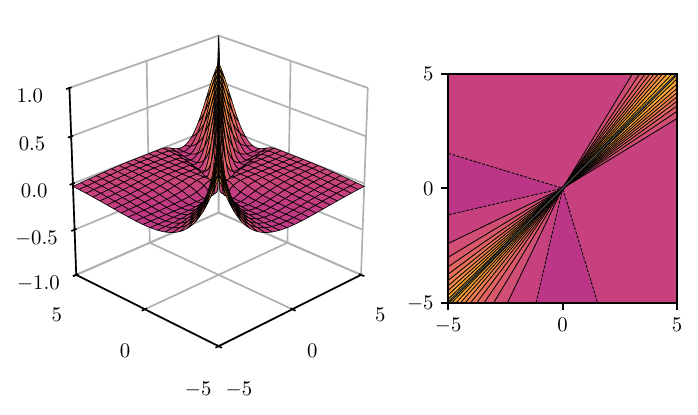}
          & \includegraphics[width=.8\linewidth, valign=m]{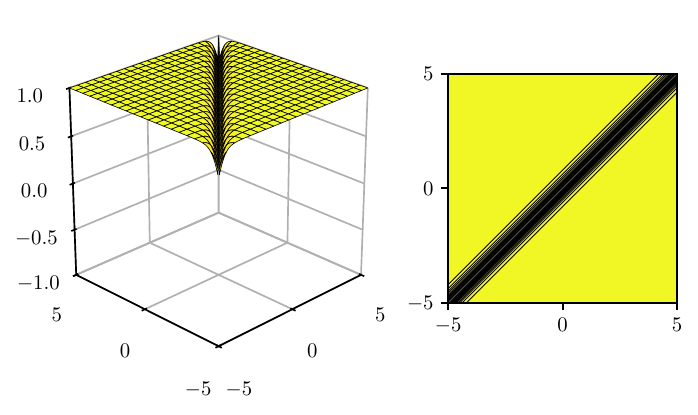}\\
      \bottomrule
    \end{tblr}%
\end{table}

Shadow-lifting has been promoted as one of the corner-stone analytic properties for differentiable logic in~\citep{varnaiRobustnessMetricsLearning2020}. Recall that DL2, STL and QLL are all shadow-lifting; and fuzzy logics are not, with the exception of the product $t$-norm (cf. \cref{tab:differentiable-logics-b}).
Fuzzy logic conjunction and disjunction exhibit poor gradient behaviour\footnote{Note that, for consistency with~\citep{vankriekenAnalyzingDifferentiableFuzzy2022}, we use linear fuzzy operators here, see~\cref{tab:differentiable-logics-b}.}.
In particular, the Gödel, \L ukasiewicz, and Yager $t$-norms (i.e. the semantics of conjunction) are not shadow-lifting; moreover, gradients of the latter two vanish on parts of their domain (cf.~\cref{tab:conjunction-disjunction}).
%
%The Yager t-norm also vanishes on a fraction of its inputs.
%While the product t-norm is shadow-lifting and only vanishes if both conjuncts are zero, its gradients are not sensibly aligned with conjunction semantics (i.e. fix most violated, 0*0).
%
In addition to that, the gradients of fuzzy implications $\phi\mulimp\psi$ as presented in~\cref{tab:differentiable-logics-a} vanish whenever $\sem{\phi}\le\sem{\psi}$.
%This makes them unsuitable as mappings for a comparison $x\le y$, since adversarial training requires informative gradients even when $x\le y$ in order to find nearby violations. % repeated above?

\emph{As \Cref{tab:main_results-long,tab:main_results-short} show, fuzzy differentiable logics consistently under-perform compared to the shadow-lifting differentiable logics, thus our experiments confirm the hypothesis of~\citep{varnaiRobustnessMetricsLearning2020}.}

However, the experiments may again also show the effects of the interval choice, as discussed in \cref{subsec:int}.
Moreover, \acs{STL} is shadow-lifting but still lags behind \acs{QLL}.
To explain this, it is interesting to compare the gradients of \acs{STL} and \acs{QLL} each with those of $\max$, as we do in \cref{tab:QLL_STL_grad_vs_max}.
We see the gradient of \ac{STL} deviate quite strongly from that of $\max$, and thus lack the good properties described in \cref{sec:qll} for \ac{QLL}.

%Together, these observations suggest that fuzzy logics are poorly suited for adversarial property-driven training.
%This is confirmed experimentally on the dice dataset with the $\mathsf{NotBoth}$ constraint, where the Gödel and \L ukasiewicz logics perform exactly like the baseline, while the Yager logic even performs slightly worse.
%Notably, the product fuzzy logic, the only fuzzy logic considered here that enjoys shadow-lifting, achieves the best performance among the fuzzy logics.

%More broadly, these results again highlight the difference between bounded and unbounded logic domains.
%Logics defined on bounded domains, i.e. \ac{DL2} and fuzzy logics, perform substantially worse than those on unbounded domains, i.e. \ac{QLL} and \ac{STL}.
%In this setting, however, the performance gap between \ac{QLL} and \ac{STL} is considerably smaller than in other experiments.

\section{Limitations and future work}\label{sec:discussion}

From a theoretical standpoint, we left open various issues.
The biggest limitation is being not able to adequately explain for the empirical behaviour of \ac{QLL}: intuitively, \ac{QLL} better manages the gap between the soft spec $\sem{\phi}$ and the hard one $\sem[\Bool]{\phi}$ because the soft connectives approximate the hard ones and due to logical adequacy.
However, we do not prove rigorous \emph{bounds} linking the test-time logical constraint loss value to the satisfaction of the hard spec.
Since verification is a much more expensive process compared to training, having guarantees on verification by evaluating the loss only would be of great practical value, evenmoreso for complex models.
% We couldn't exhibit an actual proof theory for \ac{QLL}---this is in fact subject of a paper in preparation.
% Moreover, while we limited to a propositional theory the need of a proper first-order theory emerges.
% We limited our language to only propositional variables $y_i$ and $\hat{y}_i$.
% This is fine and conceptually clear for classification tasks, in which the output of the network has a logical meaning ($f_\theta(\vec{x})$ is effectively a quantitative predicate over the set of labels), but it might not be true for other classes of tasks (e.g. generation), and does not apply to the inputs $\vec{x}$, which one might want to predicate on in the spec.

The ideas of~\citet{capucciQuantifiersQuantitativeReasoning2025} suggest this requires a first-order extension of \ac{QLL}.
DL2 is superior in this sense, having a first-order language, but not a proper first-order logic (with proof rules, model theory, etc.)
From an empirical standpoint, our experiments are purely demonstrative and cannot be considered to reflect real-world applications, based on model and dataset size as well as complexity of the logical constraint.

% It is enticing, nonetheless, to notice that first-order \ac{QLL} allows to interpret almost the \emph{entire} loss function \eqref{eq:combined_optimisation} logically, namely as
% \begin{equation}
%     \forall^1(\hat{x}, \hat{y}) \in \mathcal{D}.\ \mathcal{L}_{\mathrm{pred}}(\hat{\vec{y}}, f_\theta(\hat{\vec{x}})) \land \forall^\infty \vec{x} \in B_\varepsilon(\hat{\vec{x}}).\ \phi(\vec{x}, f_\theta(\vec{x}), \hat{\vec{x}}, f_\theta(\hat{\vec{x}})).
% \end{equation}
% Here the decorations on the quantifier stand for their `hardness', just like we have seen on the non-linear propositional connectives.
% Ignoring those, it is remarkable how the above first-order sentence reads natural: for all elements in the training set, we want to satisfy the loss and for all nearby points, the specification.

\begin{ack}
Flinkow and Monahan acknowledge the support of Taighde Éireann – Research Ireland grant number 20/FFP-P/8853 and the ADAPT Research Ireland Centre for AI Driven Digital Content Technology at Maynooth University under grant 13/RC/2106\_P2.
Komendantskaya acknowledges the partial support of the EPSRC grant AISEC: AI Secure and Explainable by Construction (EP/T026960/1),
and support of the ARIA: Mathematics for Safe AI grant.
Capucci and Komendantskaya acknowledge the support of the ARIA `Safeguarded AI' programme through project code MSAI-PR01-P05.

The authors would like to thank Gusts Gustavs Grīnbergs for his valuable contributions to and improvements of the codebase used in this work.

The authors have no competing interests to declare that are relevant to the content of this article.
\end{ack}

{
\small
\bibliographystyle{abbrvnat}
\bibliography{references}
}

\appendix

\crefalias{section}{appendix}
\crefalias{subsection}{appendix}

\section{Technical appendices and supplementary material}
\label{sec:appendix}

The appendix provides proofs, technical details, hyperparameters, and additional experiments that could not be included in the main paper.

\subsection{Residuated Lattices}\label{subsec:residuated-lattices}

\Cref{sec:qll} introduced the semantics of \ac{QLL} in $\ExtReals$ as exhibiting a `quantitative' version of a \emph{residuated lattice}.
This appendix unpacks what that means.
We reference~\citet{galatosResiduatedLatticesAlgebraic2007} and forthcoming work on the proof and model theory of \ac{QLL}~\citep{capucciQuantitativeLinearLogic2026}.

Residuated lattices are the standard algebraic semantics of substructural logics (like linear and fuzzy logic), and the laws collected in~\cref{sec:qll} are quantitative instances of laws that hold in any residuated lattice.

\paragraph{Lattices.}
A \emph{partially ordered set}, or \emph{poset}, is a pair $(L, \leq)$ where $\leq$ is reflexive ($a \leq a$), transitive ($a \leq b$ and $b \leq c$ implies $a \leq c$), and antisymmetric ($a \leq b$ implies $b \nleq a$).
Given two elements $a, b \in L$, their \emph{join} $a \vee b$ is their least upper bound and their \emph{meet} $a \wedge b$ is their greatest lower bound, when these exist.

The `least upper bound' property of the join $a \vee b$ is characterized by two clauses:
\begin{align}\label{eq:latt-upp}
  \tag{upper bound}
  c \leq a \ \text{or}\ c \leq b \quad &\text{then}\quad c \leq a \vee b,\\
  \tag{minimality}\label{eq:latt-min}
  a \leq c\ \text{and}\ b \leq c \quad &\text{then}\quad a \vee b \leq c.
\end{align}
The first says $a \vee b$ dominates each of $a$ and $b$; the second says it is the smallest such element.
The meet $a \wedge b$ is characterized dually as a greatest lower bound.

Then, a \emph{lattice} is a poset in which every pair of elements admits both a join and a meet.

Two examples of lattices relevant to this work are the extended reals $(\ExtReals, \geq)$, so that $a \vee b = \min(a, b)$ and $a \wedge b = \max(a, b)$; and the Booleans $\Bool = \{\false, \true\}$ with $\false \leq \true$, where $\vee$ is logical disjunction and $\wedge$ is logical conjunction.

The main idea behind the algebraic semantics of \ac{QLL} is to replace the order $\leq$ with a quantitative predicate instead, which is $\mulimp$ in $\ExtReals$.
This makes it a kind of `asymmetric' metric space, as observed already by~\citet{lawvereMetricSpacesGeneralized1973}.
Thus~\eqref{eq:latt-upp} and~\eqref{eq:latt-min} become the clauses appearing in~\cref{sec:qll} as the \emph{disjunction rules} of \ac{QLL}.
The same happens for the \emph{conjunction rules}, which are in fact formal duals.

\paragraph{Residuated lattices.}
A (commutative) \emph{residuated lattice} adds to a lattice a second pair of operations: a `multiplication' and an `implication'.
Concretely, it is a structure $(L, \vee, \wedge, \otimes, i, \to)$ where
\begin{enumerate}
  \item $(L, \vee, \wedge)$ is a lattice;
  \item $(L, \otimes, i)$ is a commutative monoid, i.e.\ $\otimes$ is associative, commutative, and has $i$ as a two-sided neutral element ($a \otimes i = i \otimes a = a$);
  \item (\textbf{mix}) $(a \leq b)$ and $(c \to d)$ imply $(a \otimes c) \leq (b \otimes d)$,
  \item $\to$ is the \emph{residual} of $\otimes$, meaning that the two operations are linked by the \emph{residuation property}
  \begin{equation}\label{eq:residuation}
    a \otimes b \leq c \quad\Longleftrightarrow\quad a \leq b \to c.
  \end{equation}
\end{enumerate}

Equation~\eqref{eq:residuation} pins down $\to$ uniquely given $\otimes$ and the order: $b \to c$ is the largest element $a$ such that $a \otimes b \leq c$.
%
% The residual lets one read the order off the algebra: $a \leq b$ iff $i \leq a \to b$.
% Hence a single number, the value of $a \to b$, encodes both the qualitative fact ``$a$ entails $b$'' (when it is at least $i$) and a quantitative measure of how much $a$ exceeds $b$ when this fails.
% This is exactly how \ac{QLL} uses $\to$ to encode inequalities between losses.
%
When $\otimes$ happens to coincide with the meet $\wedge$, the residual $\to$ is exactly the implication of \emph{Heyting algebras}, the algebraic semantics of intuitionistic propositional logic.
% \paragraph{Derived laws.}
% From the three axioms above one can derive the following laws, all of which appear, with the order reversed, in the table on page~\pageref{table:add-reals}:
% \begin{itemize}
%   \item \textbf{Identity:} $e \leq a \to a$.
%   \item \textbf{Modus ponens} (transitivity of implication): $(a \to b) \otimes (b \to c) \leq (a \to c)$.
%   \item \textbf{Currying:} $(a \otimes b) \to c = a \to (b \to c)$,
%   \item \textbf{Disjunction rules:} $(a \vee b) \to c = (a \to c) \wedge (b \to c)$,
%   \item \textbf{Conjunction rules:} $a \to (b \wedge c) = (a \to b) \wedge (a \to c)$,
%   \item \textbf{Distributivity of $\otimes$ over $\vee$:} $a \otimes (b \vee c) = (a \otimes b) \vee (a \otimes c)$.
%   The monoid distributes over joins (indeed, preservation of all suprema follows from~\eqref{eq:residuation}).
%   \item \textbf{Monotonicity:} both $\otimes$ and $\to$ are monotone in their arguments ($\to$ contravariantly in the first, covariantly in the second).
% \end{itemize}

$\ExtReals$ is a commutative residuated lattice with $\otimes = +$ as defined in \cref{table:add-reals}, and $\Bool$ is too with $\otimes = \land$---in fact, $\Bool$ is an Heyting algebra.

Again, in \ac{QLL} we treat the laws quantitatively.
In $\AddReals$, unitality, associativity, and commutativity of $\otimes = +$ still hold in the traditional sense, while our `mix' and `residuation' rules from~\cref{table:log-props} are a quantitative instantiation of the same rules we just introduced for residuated lattices.
In particular, note that for $\AddReals$, $\mulimp$ is both the quantitative order and the residuation $\to$ (compare the situation with $\Bool$, in which also order and residuation are given by classical material implication `$\Rightarrow$').

Often the lattice carries an order-reversing involution $a \mapsto a^*$ so that
\begin{equation}
	a \otimes b \leq c^* \iff a \leq (b \otimes c)^*
\end{equation}
This is the case for both our running examples, where $(-)^*$ is $-(-)$ for $\ExtReals$ and $\neg(-)$ for $\Bool$.
Then we can prove `contraposition':
\begin{equation}
  a \to b = b^* \to a^*.
\end{equation}

The reader interested in understanding how varieties of commutative residuated lattices form the semantics for both fuzzy logics can refer to~\citep{galatosResiduatedLatticesAlgebraic2007,metcalfeProofTheoryFuzzy2009} for fuzzy logic and to~\citep{aglianoAlgebraicInvestigationLinear2025} for linear logic.

\subsection{Proof of~\cref{thm:soundness}}\label{app:soundness}

We proceed by induction on the structure of $\phi$, recalling that at $p=\infty$ the non-linear connectives become $\sem[\infty]{\phi \Fand \psi} = \max(\sem[\infty]{\phi}, \sem[\infty]{\psi})$ and $\sem[\infty]{\phi \For \psi} = \min(\sem[\infty]{\phi}, \sem[\infty]{\psi})$.

\paragraph{Propositional variables.} For $\phi = y_i$ we have $\sem[\Bool]{y_i} = (\pi_i \leq 0) = (\sem[\infty]{y_i} \leq 0)$ by definition; the cases $\phi = \hat{y}_i$ and $\phi = r$ are identical.

\textbf{Negation.} By definition $\sem[\Bool]{\Fnot \phi} = (\sem[\infty]{\phi} \geq 0)$, which is equivalent to $-\sem[\infty]{\phi} \leq 0$, i.e.\ $\sem[\infty]{\Fnot \phi} \leq 0$.

\paragraph{Conjunction.} By induction hypothesis applied to $\phi$ and $\psi$,
\begin{equation*}
	\sem[\Bool]{\phi \Fand \psi}
	\;=\; \sem[\Bool]{\phi} \land \sem[\Bool]{\psi}
	\;\iff\; \sem[\infty]{\phi} \leq 0 \;\text{and}\; \sem[\infty]{\psi} \leq 0
	\;\iff\; \max(\sem[\infty]{\phi}, \sem[\infty]{\psi}) \leq 0,
\end{equation*}
and the right-hand side is $\sem[\infty]{\phi \Fand \psi} \leq 0$.

\paragraph{Disjunction.} Symmetrically, by induction hypothesis,
\begin{equation*}
	\sem[\Bool]{\phi \For \psi}
	\;=\; \sem[\Bool]{\phi} \lor \sem[\Bool]{\psi}
	\;\iff\; \sem[\infty]{\phi} \leq 0 \;\text{or}\; \sem[\infty]{\psi} \leq 0
	\;\iff\; \min(\sem[\infty]{\phi}, \sem[\infty]{\psi}) \leq 0,
\end{equation*}
and the right-hand side is $\sem[\infty]{\phi \For \psi} \leq 0$.

\paragraph{Implication.} By definition $\sem[\Bool]{\phi \mulimp \psi} = (\sem[\infty]{\psi} \leq \sem[\infty]{\phi})$, which is equivalent to $\sem[\infty]{\psi} - \sem[\infty]{\phi} \leq 0$, i.e.\ $\sem[\infty]{\phi \mulimp \psi} \leq 0$.
\qed.

\subsection{Differentiable Logics}\label{subsec:differentiable-logics-implemented}
\Cref{tab:differentiable-logics-a,tab:differentiable-logics-b} presents the differentiable logics implemented in our framework.
In addition to \acf{DL2}~\citep{fischerDL2TrainingQuerying2019} and the robustness metric for \acf{STL} proposed by~\citet{varnaiRobustnessMetricsLearning2020}, we also consider $t$-norm based fuzzy logics, where (a strong) conjunction is represented using a $t$-norm $T(x,y)$.

Note that fuzzy differentiable logics were introduced together with linear and non-linear logical connectives; and~\cref{tab:differentiable-logics-a,tab:differentiable-logics-b} reflect that.
However, differentiable logics originating from machine learning (\ac{DL2} and \ac{STL}) had only one set of logical connectives, which is reflected in the tables, as well.
For \ac{QLL}, we train using soft non-linear conjunction ($\land$) and disjunction ($\lor$), as well as implication.
Our hypothesis is that soft \ac{QLL} connectives are best for training as they converge to $\min$ and $\max$, which is the expected lattice interpretation of logical $\land$ and $\lor$ that ultimately gives us the Soundness result in~\cref{thm:soundness}.
For fuzzy logics, seeing that only their linear connectives are distinct, we train with linear conjunction ($\otimes$) and disjunction ($\oplus$), as well as implication.
To see how non-linear connectives of fuzzy logics will behave, it is sufficient to check the results for Gödel logic (for which linear and non-linear connectives coincide).

Note that \ac{STL} only has non-associative $n$-ary conjunction and disjunction, given in~\eqref{eq:stl-conjunction} and~\eqref{eq:stl-disjunction}, where $\nu>0$ is a constant, $\phi_{\min}=\min_i\{\dl{\phi_i}{STL}\}$, $\phi_{\max}=\max_i\{\dl{\phi_i}{STL}\}$, $\Tilde{\phi_i}^{\min}=\frac{\dl{\phi_i}{STL}-\phi_{\min}}{\phi_{\min}}$, and $\Tilde{\phi_i}^{\max}=\frac{\dl{\phi_i}{STL}-\phi_{\max}}{\phi_{\max}}$:
\begin{equation}\label{eq:stl-conjunction}
    \dl{\textstyle \bigwedge_i \phi_i}{STL} =
    -\begin{cases}
        \dfrac{\sum_i \phi_{\min} e^{\Tilde{\phi_i}^{\min}} e^{\nu \Tilde{\phi_i}^{\min}}}{\sum_i e^{\nu \Tilde{\phi_i}^{\min}}}, & \phi_{\min}<0,\\[1ex]
        \dfrac{\sum_i \dl{\phi_i}{STL} e^{-\nu \Tilde{\phi_i}^{\min}}}{\sum_i e^{-\nu \Tilde{\phi_i}^{\min}}}, & \phi_{\min}>0,\\[1ex]
        0, & \phi_{\min}=0.
    \end{cases}
\end{equation}
\begin{equation}\label{eq:stl-disjunction}
    \dl{\textstyle \bigvee_i \phi_i}{STL} =
    -\begin{cases}
        \dfrac{\sum_i \phi_{\max} e^{\Tilde{\phi_i}^{\max}} e^{\nu \Tilde{\phi_i}^{\max}}}{\sum_i e^{\nu \Tilde{\phi_i}^{\max}}}, & \phi_{\max}>0,\\[1ex]
        \dfrac{\sum_i \dl{\phi_i}{STL} e^{-\nu \Tilde{\phi_i}^{\max}}}{\sum_i e^{-\nu \Tilde{\phi_i}^{\max}}}, & \phi_{\max}<0,\\[1ex]
        0, & \phi_{\max}=0.
    \end{cases}
\end{equation}

Note also we use $\dl{\phi}{STL} = -\rho$ where $\rho$ is their robustness metric because \citep{varnaiRobustnessMetricsLearning2020} use the opposite order from us, maximising rather than minimising the resulting semantics.

\begin{landscape}
%\vspace*{\fill}
\footnotesize
\centering
\begin{talltblr}
[
    theme=paper,
    caption={An overview of differentiable logics considered in this paper. $\dl{-}{}$ denotes the mapping of a logical formula into loss. Connectives used in our experiments are highlighted in blue.},
    label={tab:differentiable-logics-a},
    note{a}={For $r=1$, the Yager t-norm corresponds to the \L ukasiewicz t-norm, and it approaches the Gödel t-norm as $r\to\infty$.},
    note{b}={For \acs{DL2}, we use it's atomic comparison operator $\dl{\phi\le\psi}{DL2}=\max(\dl{\phi}{DL2}-\dl{\psi}{DL2}, 0)$ for implication, and for \acs{STL}, we use $\dl{\phi\mulimp\psi}{STL}=\dl{\psi}{STL}-\dl{\phi}{STL}$.},
    note{c}={Atoms in \acs{DL2} are $\dl{t\le t'}{DL2}$ and $\dl{t\neq t'}{DL2}$. Instead of providing a dedicated negation operator, negation is pushed inwards to the level of comparison.},
]
{
    colspec={
        X[l,m,mode=text]
        *{3}{Q[c,m,mode=math]}
        *{2}{Q[c,m,mode=math]}
    },
    hlines, vlines,
    width=\linewidth,
    row{1}={mode=text,font=\bfseries},
    %rowsep=2pt,
    cell{2-8}{5}={bg=blue!20},
}
%\toprule
Logic
& Domain
& $\dl{\top}{}$
& $\dl{\bot}{}$
%& $\dl{t\le t'}{}$
& $\dl{\phi \mulimp \psi}{}$
& $\dl{\lnot \phi}{}$
\\ %\midrule

Gödel \citep{vankriekenAnalyzingDifferentiableFuzzy2022,affeldtFoundationDifferentiableLogics2026}
& \SetCell[r=4,c=1]{c} [0,1]
& \SetCell[r=4,c=1]{c} 1
& \SetCell[r=4,c=1]{c} 0
%& \SetCell[r=4,c=1]{c} \max\left(1-\max\left(\frac{\dl{t}{}-\dl{t'}{}}{\dl{t}{}+\dl{t'}{}}, 0\right), 0\right)
& \begin{cases}1,&\text{if }\dl{\phi}{G}\le\dl{\psi}{G},\\\dl{\psi}{G},&\text{otherwise}.
\end{cases}
& \begin{cases}1,&\text{if }\dl{\phi}{G}=0,\\0,&\text{otherwise}.\end{cases}
\\

\L ukasiewicz  \citep{vankriekenAnalyzingDifferentiableFuzzy2022,affeldtFoundationDifferentiableLogics2026}
& 
& 
& 
%& 
& \min(1-\dl{\phi}{\text{\L}}+\dl{\psi}{\text{\L}}, 1) & 1-\dl{\phi}{\text{\L}}
\\

Product \citep{vankriekenAnalyzingDifferentiableFuzzy2022,affeldtFoundationDifferentiableLogics2026}
& 
& 
& 
% & 
& \begin{cases}1,&\text{if }\dl{\phi}{P}\le\dl{\psi}{P},\\\frac{\dl{\psi}{P}}{\dl{\phi}{P}},&\text{otherwise}\end{cases} & \begin{cases}1,&\text{if }\dl{\phi}{G}=0,\\0,&\text{otherwise}.\end{cases}
\\

${}_r$Yager{\TblrNote{a}}\hspace{.5em} \citep{vankriekenAnalyzingDifferentiableFuzzy2022,affeldtFoundationDifferentiableLogics2026}
& 
& 
& 
%& 
& \begin{cases}1,\text{if } \dl{\phi}{Y}\le\dl{\psi}{Y},\\1-((1-\dl{\psi}{Y})^r -(1-\dl{\phi}{Y})^r)^{\nicefrac{1}{r}},&\text{otherwise.}\end{cases} & (1-(1-(1-\dl{\phi}{Y}))^r)^{\nicefrac{1}{r}}
\\

%\midrule

\acs{DL2} \citep{fischerDL2TrainingQuerying2019}
& [0,\infty]
& 0
& \infty
% & \max(0,t-t')
& \text{not provided}\TblrNote{b}
& \text{not provided}\TblrNote{c}
\\

%\midrule

${}_{\nu}$\acs{STL} \citep{varnaiRobustnessMetricsLearning2020}
& [-\infty,\infty]
& \infty
& -\infty
% & t'-t
& \text{not provided}\TblrNote{b}
& -\dl{\phi}{STL}
\\

%\midrule

${}_p$\acs{QLL} (this paper)
& [-\infty,\infty]
& -\infty
& \infty
% & t-t'
& \dl{\psi}{QLL} - \dl{\phi}{QLL}
& -\dl{\phi}{QLL}
\\

%\bottomrule
\end{talltblr}
\par\vspace{\fill}\par
\footnotesize
\centering
\begin{talltblr}
[
    theme=paper,
    caption={An overview of differentiable logics considered in this paper. $\dl{\land}{}$ and $\dl{\lor}{}$ denote the non-linear connectives, and $\dl{\otimes}{}$ and $\dl{\oplus}{}$ denote the linear ones. Connectives used in our experiments are highlighted in blue. Their plots and gradients are shown in~\cref{tab:conjunction-disjunction}.},
    label={tab:differentiable-logics-b},
    note{a}={\acs{DL2} and \acs{STL} do not admit a meaningful distinction between non-linear and linear connectives.},
]
{
    colspec={
        X[l,m,mode=text]
        *{4}{Q[c,m,mode=math]}
    },
    hlines, vlines,
    width=\linewidth,
    row{1}={mode=text,font=\bfseries},
    %rowsep=2pt,
    cell{2-5}{3,5}={bg=blue!20},
    cell{6}{2-5}={bg=blue!20},
    cell{6,7}{2-5}={bg=blue!20},
    cell{8}{2,4}={bg=blue!20},
}
%\toprule
Logic
& $\dl{\phi\land \psi}{}$
& $\dl{\phi \otimes \psi}{}$
& $\dl{\phi\lor \psi}{}$
& $\dl{\phi \oplus \psi}{}$
\\ %\midrule

\text{Gödel}
\citep{vankriekenAnalyzingDifferentiableFuzzy2022,affeldtFoundationDifferentiableLogics2026}
& \SetCell[r=4,c=1]{c} \min(\dl{\phi}{},\dl{\psi}{})
& \min(\dl{\phi}{G}, \dl{\psi}{G})
& \SetCell[r=4,c=1]{c} \max(\dl{\phi}{},\dl{\psi}{})
& \max(\dl{\phi}{G}, \dl{\psi}{G})
\\

\text{\L ukasiewicz}
\citep{vankriekenAnalyzingDifferentiableFuzzy2022,affeldtFoundationDifferentiableLogics2026}
&
& \max(\dl{\phi}{\text{\L}}+\dl{\psi}{\text{\L}}-1,0)
&
& \min(\dl{\phi}{\text{\L}}+\dl{\psi}{\text{\L}},1)
\\

Product \citep{vankriekenAnalyzingDifferentiableFuzzy2022,affeldtFoundationDifferentiableLogics2026}
&
& \dl{\phi}{P}\dl{\psi}{P}
&
& \dl{\phi}{P}+\dl{\psi}{P}-\dl{\phi}{P}\dl{\psi}{P}
\\

${}_r$Yager \citep{vankriekenAnalyzingDifferentiableFuzzy2022,affeldtFoundationDifferentiableLogics2026}
&
& \max(1-((1-\dl{\phi}{Y})^r+(1-\dl{\psi}{Y})^r)^{\nicefrac{1}{r}},0)
&
& \min(\dl{\phi}{Y}^r+\dl{\psi}{Y}^r, 1)
\\

%\midrule

\acs{DL2} \citep{fischerDL2TrainingQuerying2019}
& \SetCell[r=1,c=2]{c} \dl{\phi}{DL2}+\dl{\phi}{DL2}\TblrNote{a}
& 
& \SetCell[r=1,c=2]{c} \dl{\phi}{DL2}\dl{\phi}{DL2}\TblrNote{a}
&
\\

%\midrule

${}_{\nu}$\acs{STL} \citep{varnaiRobustnessMetricsLearning2020}
& \SetCell[r=1,c=2]{c} \text{see~\eqref{eq:stl-conjunction}}\TblrNote{a}
& 
& \SetCell[r=1,c=2]{c} \text{see~\eqref{eq:stl-disjunction}}\TblrNote{a}
&
\\

%\midrule

${}_p$\acs{QLL} (this paper)
& \dl{\phi}{QLL} \pAand \dl{\psi}{QLL}
& \dl{\phi}{QLL} + \dl{\psi}{QLL}
& \dl{\phi}{QLL} \pAor \dl{\psi}{QLL}
& \dl{\phi}{QLL} +^* \dl{\psi}{QLL}
\\

%\bottomrule
\end{talltblr}
%\vspace*{\fill}
\end{landscape}

In~\cref{tab:conjunction-disjunction} we plot all the conjunction and disjunction operators surveyed here with their gradients.
Fuzzy logics are plotted with $0 \leq x,y \leq 1$ and the result is also plotted on a $0$ to $1$ scale.
The rest of the logics are plotted on $-2 \leq x,y \leq +2$ to better display the behaviour of the operations.
In~\cref{tab:QLL_STL_conjunction} we also separately plot \acs{STL} and \acs{QLL} conjunctions for varying parameters $\nu$ and $p$, respectively, along with their gradients.

\begin{table}
    \small
    \centering
    \caption{Surface and contour plots with gradients of conjunction and disjunction operators listed in~\cref{tab:differentiable-logics-b}.}
    \label{tab:conjunction-disjunction}
    \begin{tblr}
    {
      width=1.05\textwidth,
      colspec={Q[l, valign=m, mode=text, cmd=\rotatebox{90},wd=.5cm]X[2, c, valign=m, mode=text]X[c, valign=m, mode=text]X[2, c, valign=m, mode=text]X[c, valign=m, mode=text]},
      row{1}={mode=text, font=\bfseries},
      cell{1}{1}={cmd=, valign=m}
    }
      \toprule
        Logic & \SetCell[r=1, c=2]{c} Conjunction & & \SetCell[r=1, c=2]{c} Disjunction & \\
      \cmidrule[lr]{2-3} \cmidrule[lr]{4-5}
        & Surface & Gradient & Surface & Gradient \\
      \midrule
        Gödel
        & \includegraphics[width=4.5cm, valign=m]{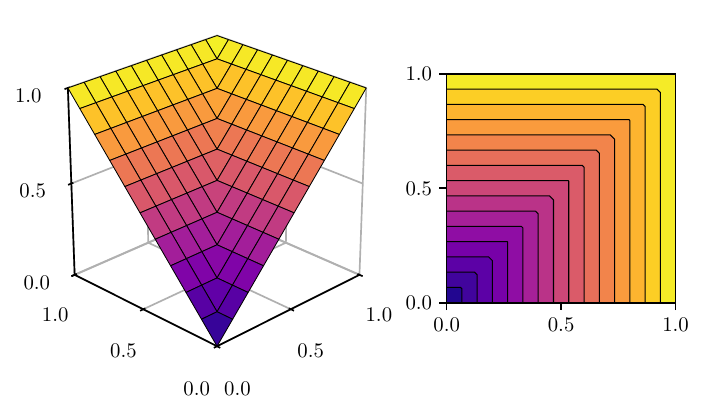} 
        & \includegraphics[width=2cm, valign=m]{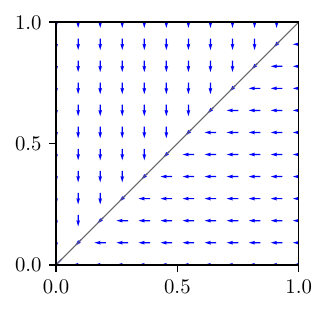}
        & \includegraphics[width=4.5cm, valign=m]{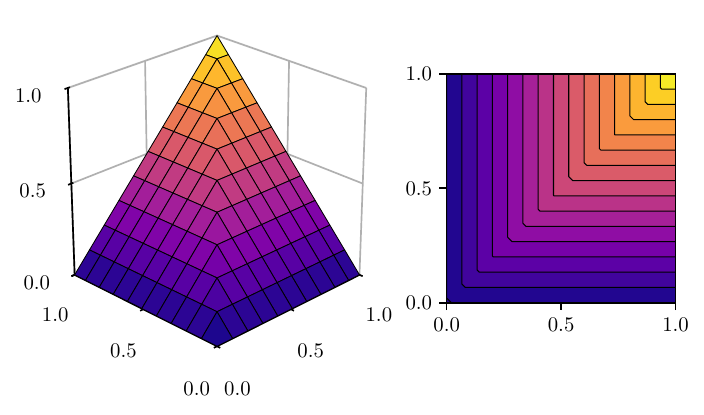} 
        & \includegraphics[width=2cm, valign=m]{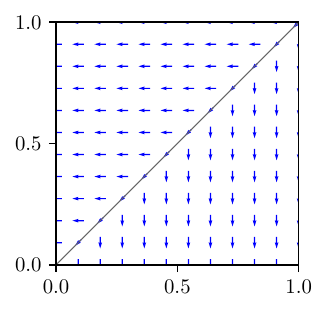}
        \\
        \L ukasiewicz
        & \includegraphics[width=4.5cm, valign=m]{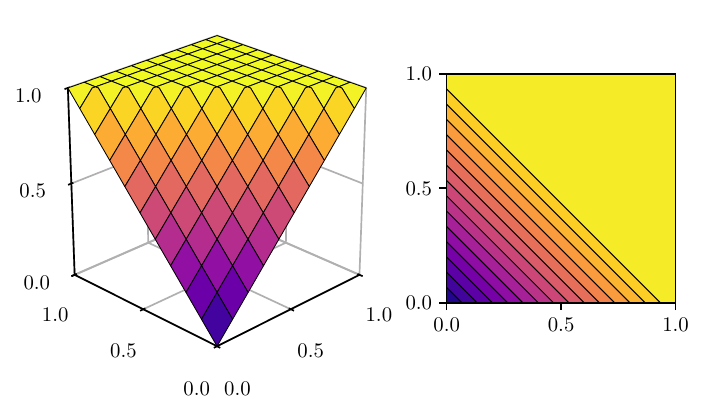} 
        & \includegraphics[width=2cm, valign=m]{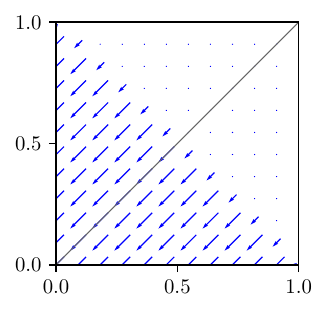}
        & \includegraphics[width=4.5cm, valign=m]{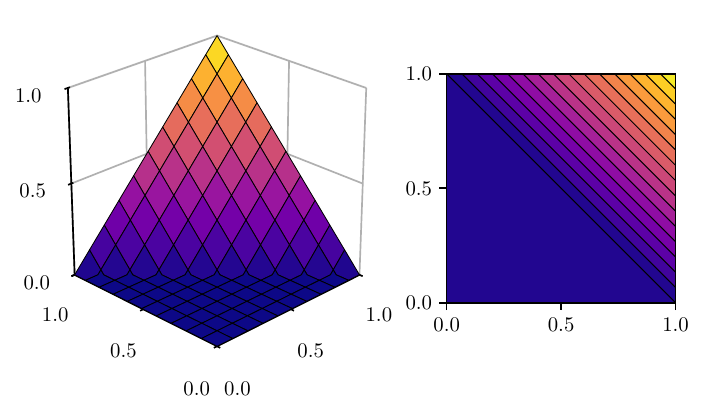} 
        & \includegraphics[width=2cm, valign=m]{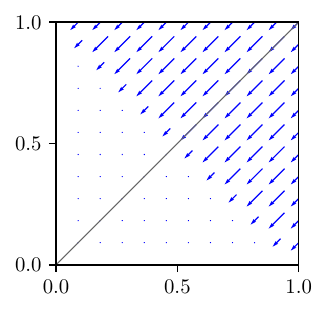}
        \\
        Product
        & \includegraphics[width=4.5cm, valign=m]{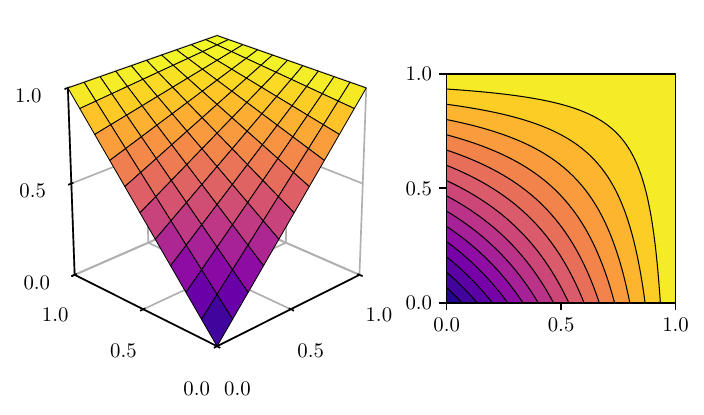} 
        & \includegraphics[width=2cm, valign=m]{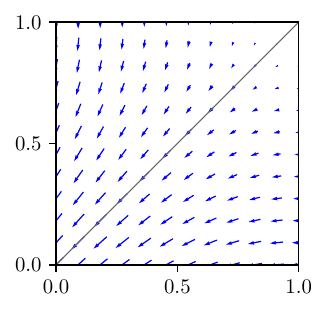}
        & \includegraphics[width=4.5cm, valign=m]{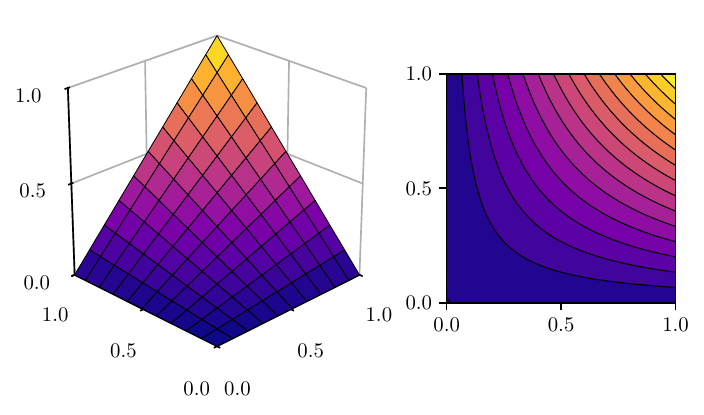} 
        & \includegraphics[width=2cm, valign=m]{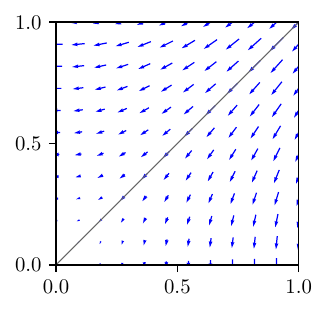}
        \\
        ${}_5$Yager
        & \includegraphics[width=4.5cm, valign=m]{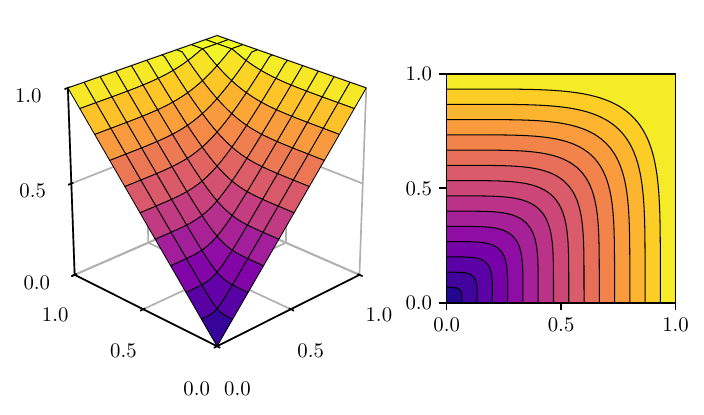} 
        & \includegraphics[width=2cm, valign=m]{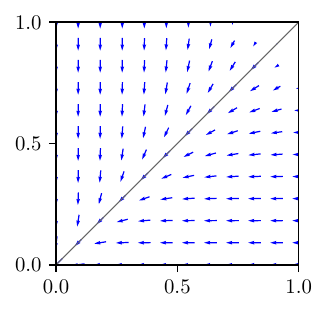}
        & \includegraphics[width=4.5cm, valign=m]{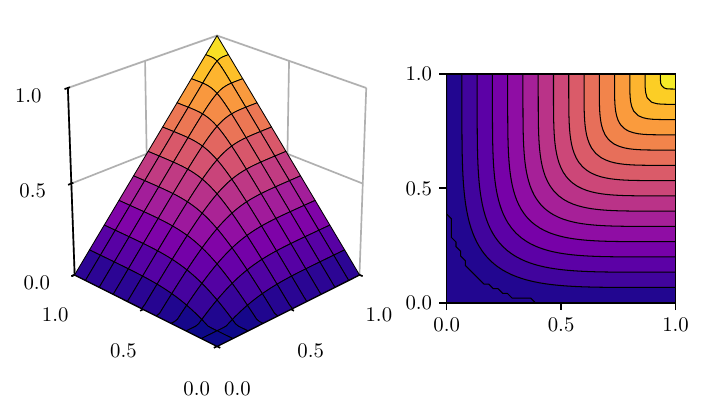} 
        & \includegraphics[width=2cm, valign=m]{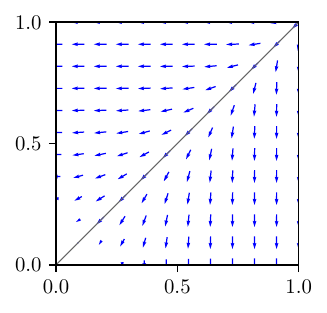}
        \\
      \midrule
        \acs{DL2}
        & \includegraphics[width=4.5cm, valign=m]{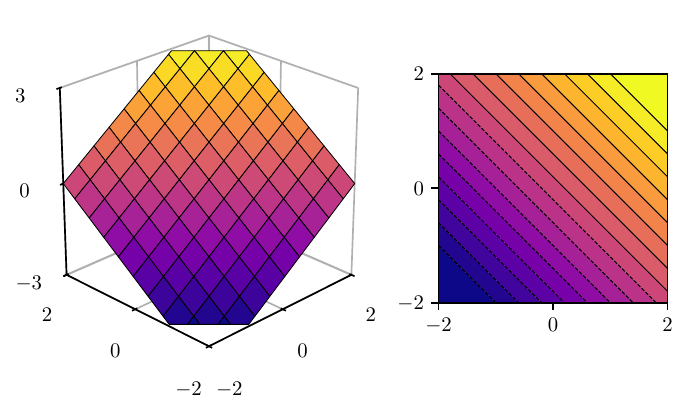} 
        & \includegraphics[width=2cm, valign=m]{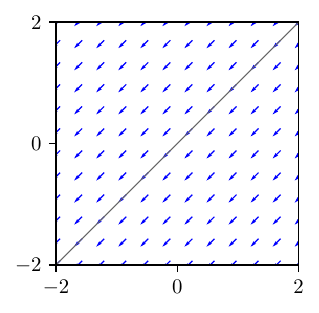}
        & \includegraphics[width=4.5cm, valign=m]{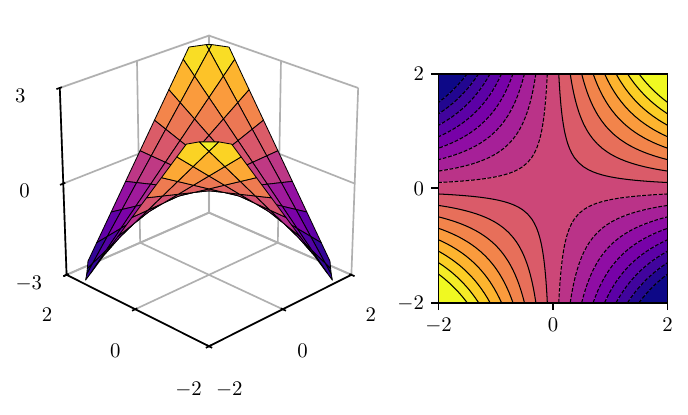} 
        & \includegraphics[width=2cm, valign=m]{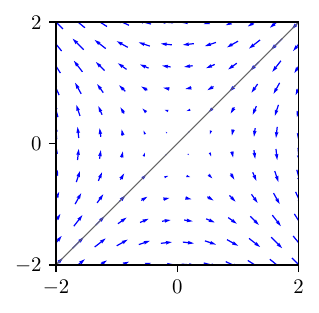}
        \\
      \midrule
        ${}_1$\acs{STL}
        & \includegraphics[width=4.5cm, valign=m]{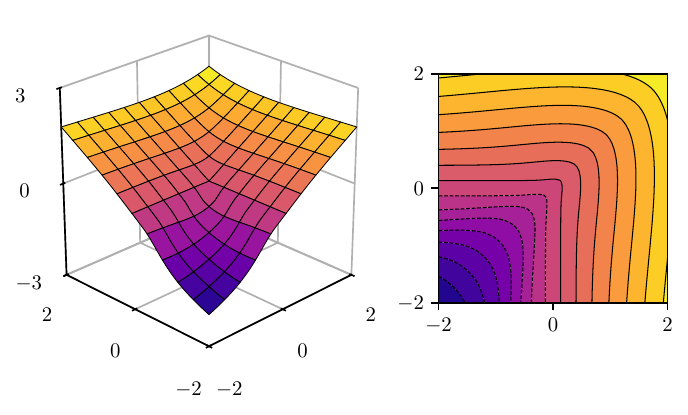} 
        & \includegraphics[width=2cm, valign=m]{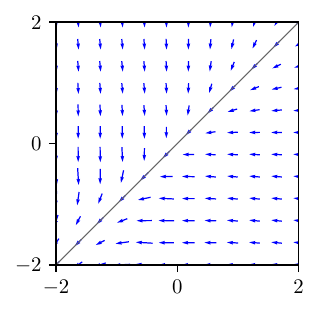}
        & \includegraphics[width=4.5cm, valign=m]{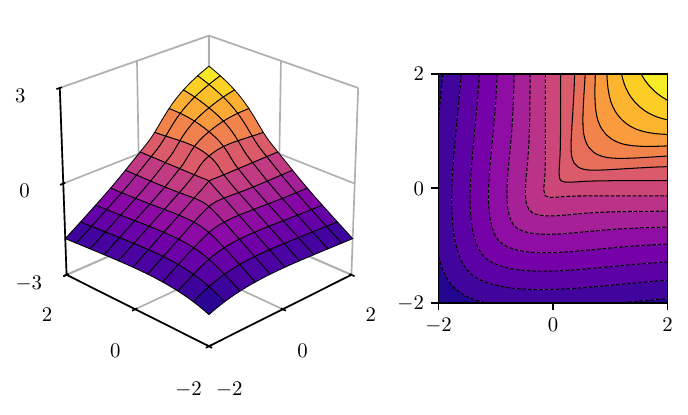} 
        & \includegraphics[width=2cm, valign=m]{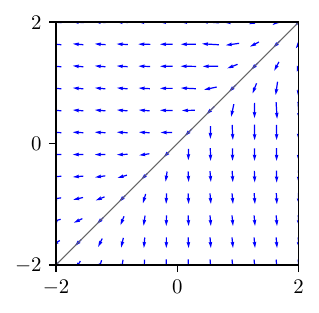}
        \\
      \midrule
        ${}_1$\acs{QLL}
        & \includegraphics[width=4.5cm, valign=m]{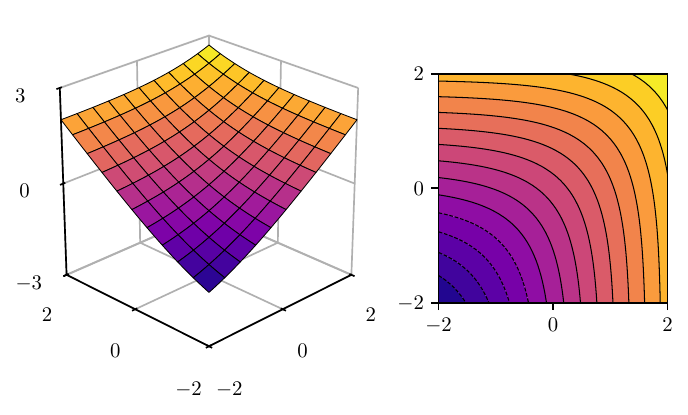} 
        & \includegraphics[width=2cm, valign=m]{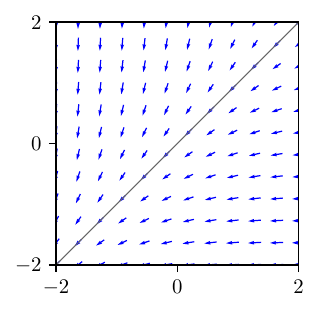}
        & \includegraphics[width=4.5cm, valign=m]{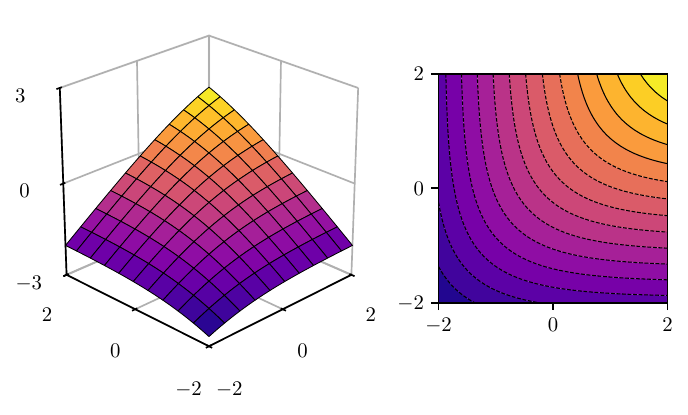} 
        & \includegraphics[width=2cm, valign=m]{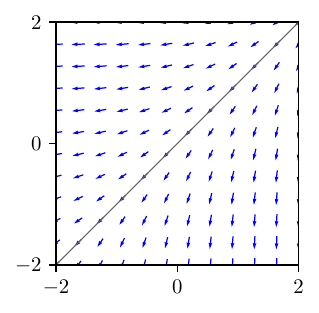}
        \\
      \bottomrule
    \end{tblr}%
\end{table}
\begin{table}
    \small
    \centering
    \caption{Surface and contour plots with gradients of the \acs{STL} conjunction operator for $\nu\in\{1,2,10\}$ and the \acs{QLL} conjunction operator for $p\in\{1,2,10\}$.
    Notice the convergence to the plot and gradients of $\max$ (first column, first row in \cref{tab:conjunction-disjunction}) for both, as expected.
    Observe also how \acs{STL} contour lines are slightly non-convex, which leads to gradients which are slightly misaligned to those of $\max$ in the central part of the plot, as seen in \cref{tab:QLL_STL_grad_vs_max}.}
    \label{tab:QLL_STL_conjunction}
    \begin{tblr}
    {
      width=\textwidth,
      colspec={X[2, c, mode=text]X[c, mode=text]X[2, c, mode=text]X[c, mode=text]},
      row{1}={mode=text, font=\bfseries},
    }
      \toprule
        \SetCell[r=1, c=2]{c} \acs{STL} & & \SetCell[r=1, c=2]{c} \acs{QLL} & \\
      \cmidrule[lr]{1-2} \cmidrule[lr]{3-4}
        Surface & Gradient & Surface & Gradient \\
      \midrule
        \includegraphics[width=4.5cm, valign=m]{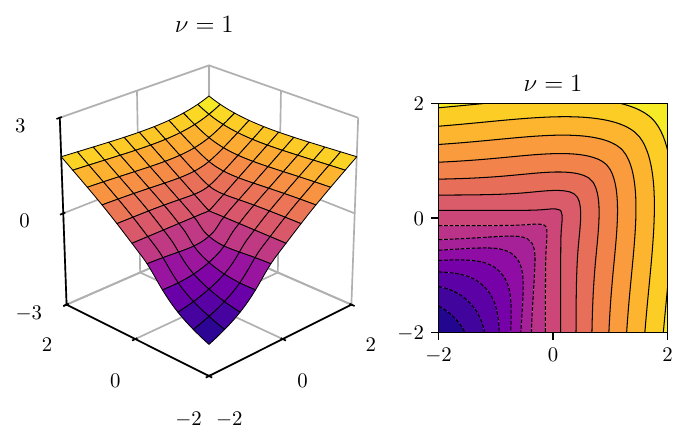} & \includegraphics[width=2cm, valign=m]{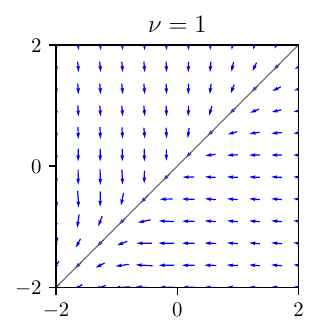} & \includegraphics[width=4.5cm, valign=m]{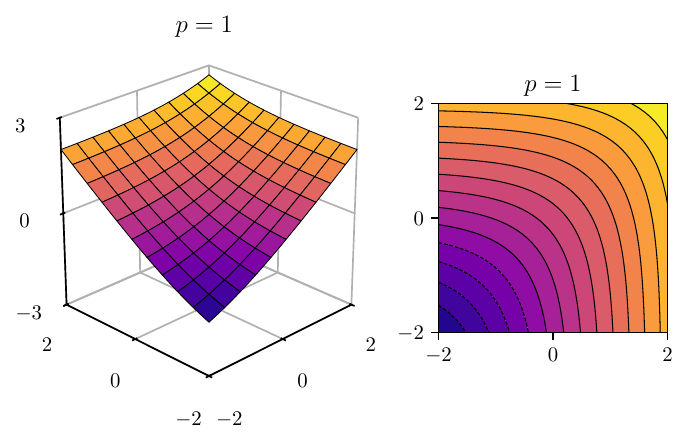} & \includegraphics[width=2cm, valign=m]{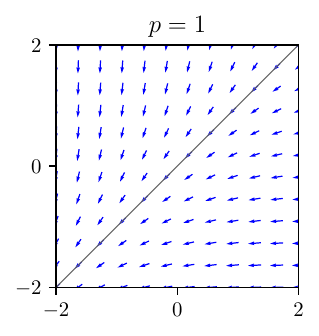}\\
        \includegraphics[width=4.5cm, valign=m]{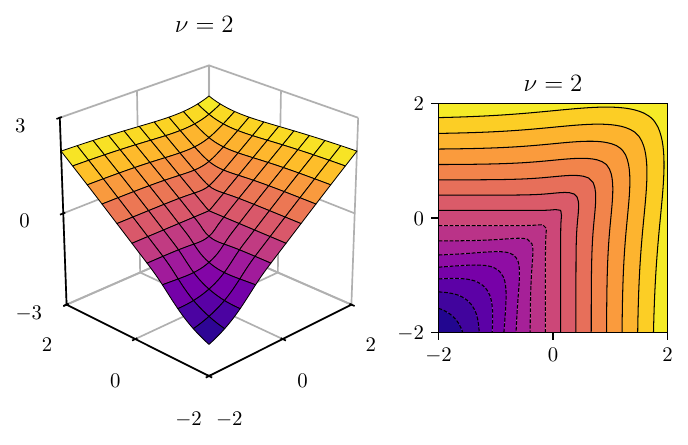} & \includegraphics[width=2cm, valign=m]{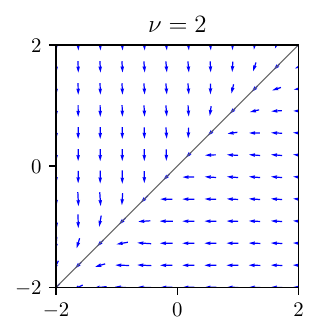} & \includegraphics[width=4.5cm, valign=m]{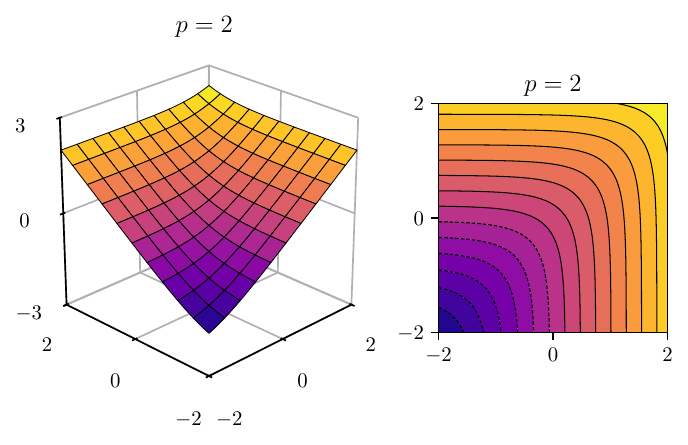} & \includegraphics[width=2cm, valign=m]{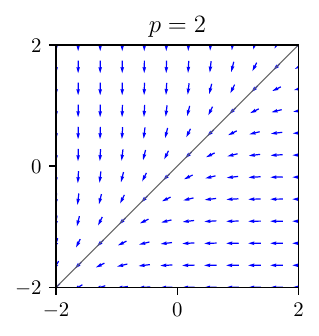}\\
        \includegraphics[width=4.5cm, valign=m]{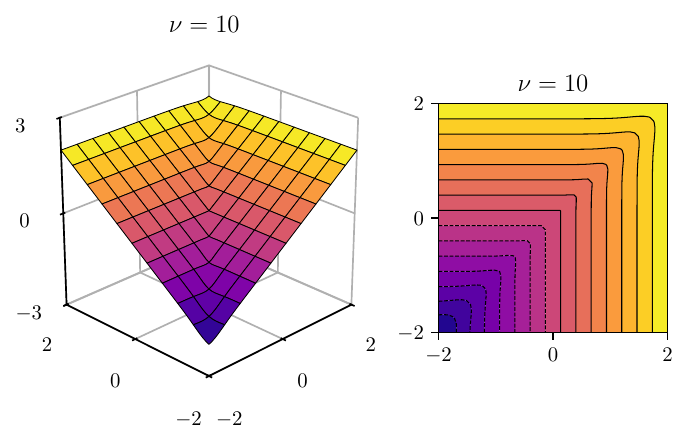} & \includegraphics[width=2cm, valign=m]{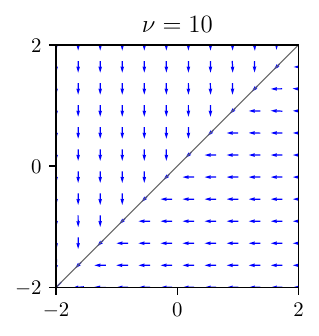} & \includegraphics[width=4.5cm, valign=m]{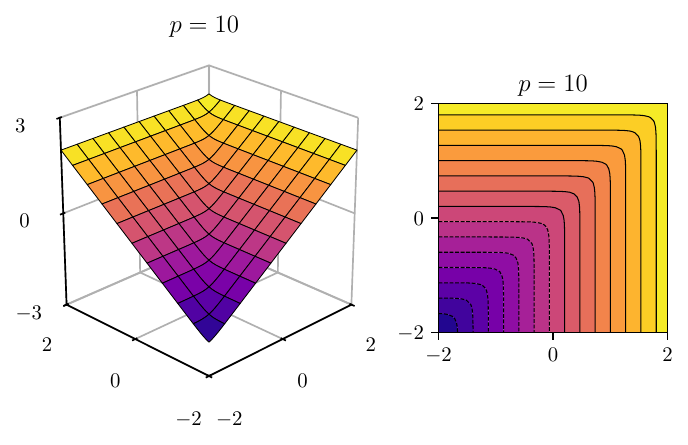} & \includegraphics[width=2cm, valign=m]{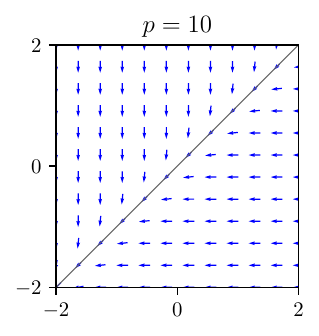}\\
      \bottomrule
    \end{tblr}%
\end{table}

\subsection{\ac{QLL} Formalisations of Properties}\label{subsec:qll-property-formalisations}

\paragraph{Strong Classification Robustness.}
In the following, $m=28 \times 28$ and $n = 10$ for MNIST and Fashion-MNIST.
For convenience, we also replaced data points $(\hat{\vec{x}}, c)$ with $(\hat{\vec{x}}, \hat{\vec{y}})$ where $\hat{\vec{y}}$ is the logits of the one-hot distribution corresponding to the true class.

\emph{Strong classification robustness}~\citep{fischerDL2TrainingQuerying2019,casadioNeuralNetworkRobustness2022} requires that the logit corresponding to the true class must be predicted with degree above some threshold $\delta$, informally:
\begin{equation}\label{eq:informal-strong-classification-robustness}
	\text{for all $\vec{x} \in B_\varepsilon(\hat{\vec{x}})$},\ f_\theta(\vec{x})_{c} \geq \delta.
\end{equation}
Note that $c$, in fact, depends on $\hat{\vec{x}}$, and thus one cannot be translated as a valid formula in the grammar we described in~\cref{def:qll-syntax}.
However, in \ac{QLL} one can instead use the following formula, where $L = \{0, \ldots, 9\}$ is the set of MNIST/Fashion-MNIST labels:
\begin{equation}\label{eq:strong-classification-robustness}
	\text{Strong classification robustness} :=\ \bigwedge_{c \in L} \hat{y}_c \mulimp (\delta \mulimp y_c).
\end{equation}
The formula above thus says `for every class $c$, if it is the true class for the data point under consideration, then the network must predict $c$ with degree at least $\delta$'.

As a sanity-check, the additive semantics of \eqref{eq:strong-classification-robustness} at $\hat{\vec{x}}$ equals the semantics of $\delta \mulimp y_c$ for $c$ the true class of $\hat{\vec{x}}$:
\begin{equation}\label{eq:scr-simplified}
	\sem{\text{Strong classification robustness}}(\vec{y}, \hat{\vec{y}}) = \sem{\delta \mulimp y_{c}}(\vec{y}, \hat{\vec{y}}).
\end{equation}
We call this the `simplified' formula.
In this form, it can be handled by \ac{DL2} and \ac{STL} which lack an implication operator but have a comparison one---in the simplified formula, we can indeed interpret $\mulimp$ as $\geq$.

\paragraph{Classification Robustness.}
\emph{Classification robustness}~\citep{casadioNeuralNetworkRobustness2022} is a stronger notion of robustness than strong classification robustness, as the latter only imposes an arbitrary bound on a single logit and does not prevent label changes.
Classification robustness requires that the logit corresponding to the true class should remain the largest under all allowed perturbations, in symbols:
\begin{equation}
	\text{for all $\vec{x} \in B_\varepsilon(\hat{\vec{x}})$},\ \text{for all $i \in L$}, f_\theta(\vec{x})_{c} \geq f_\theta(\vec{x})_i
\end{equation}
In \ac{QLL}, this is formalized as
\begin{equation}\label{eq:classification-robustness}
	\text{Classification robustness} :=\ \bigwedge_{c \in L} \left(\hat{y}_c \mulimp \bigwedge_{i \in L} (y_c \mulimp y_i)\right)
\end{equation}
Similar considerations as above apply for the translation of this to \ac{DL2} and \ac{STL}, with the additional caveat that for \ac{STL} we cannot use $\bigwedge_{i \in L} (y_c \geq y_i)$ due to ill-behaved semantics when one of the conjuncts is $0$, which always occurs for $c=i$.
%Thus for STL, we must tweak the formula to be $\bigwedge_{i \in L\setminus\{c\}} y_c \geq y_i$.

\paragraph{Semantic Relations on Fashion-MNIST.}
A more interesting constraint concerning semantic relations between classes can be encoded on Fashion-MNIST.
Here we partition the labels $L$ into sets $C=\{0,2,3,4,6\}$ for \emph{clothing} items (t-shirt/top, pullover, dress, coat, shirt) and $F=\{5,7,9\}$ for \emph{footwear} (sandal, sneaker, ankle boot).
Then, given a data point $(\hat{\vec{x}}, c)$ with true class $c$, our specification is to require the logit corresponding to the ground-truth to exceed every logit of the opposite class, informally:
\begin{equation}
 	\begin{dcases}
 		\text{for all $i \in F$,}\ f_\theta(\vec{x})_{c} \geq f_\theta(\vec{x})_i & \text{if $c \in C$}\\
 		\text{for all $i \in C$,}\ f_\theta(\vec{x})_{c} \geq f_\theta(\vec{x})_i & \text{if $c \in F$}.
 	\end{dcases}
\end{equation}
We can formalize this in \ac{QLL} as follows:
\begin{equation}\label{eq:clothing-footwear}
	\text{Clothing/Footwear} := \left( \bigwedge_{c \in C} \left(\hat{y}_c \mulimp \bigwedge_{i\in F} y_c \mulimp y_i\right) \right) \land \left( \bigwedge_{c \in F} \left(\hat{y}_c \mulimp \bigwedge_{i\in C} y_c \mulimp y_i\right) \right).
\end{equation}

\subsection{Hyperparameters and Experimental Setup}\label{subsec:additional-experimental}
All experiments were trained for $100$ epochs using the AdamW~\citep{loshchilovDecoupledWeightDecay2018} optimiser with a learning rate of \num{1e-3} and weight decay of \num{1e-4}.
In each experiment, all runs were trained under identical fixed hyperparameters, i.e., learning rate, adversarial attack configurations, and loss balancing.
For attacks, we make use of AutoPGD~\citep{croceReliableEvaluationAdversarial2020}, since the effectiveness of standard \ac{PGD} has been shown to strongly rely on good attack parameters~\citep{mosbachLogitPairingMethods2019,croceScalingRandomizedGradientFree2020}.

The experiments were run on a machine with a NVIDIA RTX 4090 with \qty{24}{\giga\byte} of VRAM, an Intel i9-13900KF CPU, and \qty{64}{\giga\byte} of RAM with and additional \qty{250}{\gibi\byte} of swap memory available for formal verification workloads.
Training required approximately \qty{3}{\day} \qty{17}{\hour} of compute time, while formal verification required approximately \qty{8}{\day} \qty{19}{\hour} of compute time.

We combine standard prediction loss $\mathcal{L}_{\mathrm{pred}}$ with the constraint loss $\sem{\phi}$ using the joint optimisation objective from~\cref{eq:combined_optimisation}.
To balance the two objectives, we dynamically set the weighting parameter $\lambda$ in~\cref{eq:combined_optimisation} based on their gradient norms, inspired by multi-task learning techniques such as~\citep{chenGradNormGradientNormalization2018}:
\begin{equation}\label{eq:grad_norm}
	\lambda=\alpha\frac{\lVert\nabla_{\theta}\mathcal{L}_{\mathrm{pred}} \rVert}{\lVert\nabla_{\theta}\sem{\phi} \rVert},
\end{equation}
where $\alpha$ controls the relative influence of the constraint objective.
We fix $\alpha=0.5$ across all experiments to ensure that both objectives contribute comparable optimisation signal independent of the underlying differentiable logic.
We do not perform per-logic hyperparameter tuning.

Since neural network verification suffers from severe scalability issues~\citep{katzReluplexEfficientSMT2017}, we intentionally use compact model architectures (c.f.~\cref{subsec:models-datasets}) with a small number of non-linearities to keep both training and verification computationally tractable.
The goal of our experiments is therefore not to achieve state-of-the-art prediction accuracy, but to study the behaviour of various differentiable logics in a verification-compatible setting.
Nevertheless, we evaluate differentiable logics based on jointly achieving high prediction accuracy and verified constraint satisfaction throughout the experiments (c.f.~\cref{tab:main_results-short,tab:main_results-long}).

\Cref{tab:hyperparameters} displays the hyperparameters used in the training experiments.
\begin{table}
    \small
    \centering
    \caption{Hyperparameters used for the experiments.}
    \label{tab:hyperparameters}
    \begin{tblr}
    {
      width=\textwidth,
      colspec={Q[l, m, mode=text]X[l, m, mode=text]Q[c, m, mode=math]Q[c, m, mode=math]Q[c, m, mode=math]},
      row{1}={mode=text, font=\bfseries},
      row{2}={mode=text},
    }
      \toprule
        \SetCell[r=2]{l} Dataset & \SetCell[r=2]{l} Constraint & \SetCell[r=2]{l} Constraint Parameters & \SetCell[c=2]{c} \acs{PGD}\\
      \cmidrule[lr]{4-5}
        & & & steps & restarts\\
      \midrule
        \SetCell[r=2]{l} MNIST & Strong classification robustness & \epsilon=0.1,\delta=0.7 & \SetCell[r=2]{c} 20 & \SetCell[r=2]{c} 2\\
        & Classification robustness & \epsilon=0.2 & & \\
      \midrule
        \SetCell[r=3]{l} Fashion-MNIST & Strong classification robustness & \epsilon=0.2,\delta=0.6 & \SetCell[r=3]{c} 20 & \SetCell[r=3]{c} 2\\
        & Classification robustness & \epsilon=0.2 & & \\
        & Clothing/Footwear & \varepsilon=0.1 & &\\
        \midrule
        \SetCell[r=2]{l} Dice & Not-Both & \SetCell[r=2]{c} \epsilon=\nicefrac{4}{255} & \SetCell[r=2]{c} 30 & \SetCell[r=2]{c} 3\\
        & Exactly-One & & &\\
      \bottomrule
    \end{tblr}%
\end{table}

\subsection{Evaluation Metrics}\label{subsec:evaluation-metrics}
Following~\citet{casadioNeuralNetworkRobustness2022}, we briefly recap the formal definitions of the three different measures we use to evaluate satisfaction of constraints:
\begin{inparaenum}[(i)]
    \item Does the constraint hold? This is a binary measure, and the answer is either true or false.
    \item If the constraint does not hold, how easy is it for an attacker to find a violation?
    \item If the constraint does not hold, how often does the average user encounter a violation?
\end{inparaenum}
Based on these measures, we define three concrete metrics:
\emph{constraint satisfaction}, \emph{constraint security}, and \emph{constraint accuracy}.

Let $\mathcal{X}$ be the test set, $B_\varepsilon(\hat{\vec{x}}) = \{\vec{x}\in\real^m \mid \lVert\hat{\vec{x}}-\vec{x}\rVert_\infty\le\varepsilon\}$ be the $\epsilon$-ball around $\hat{\vec{x}}$, and $P$ be the right-hand side of the implication of a constraint $\phi$ of the form $\forall \vec{x}\in B_\varepsilon(\hat{\vec{x}})\ldotp P(\vec{x})$.
Let $\mathbbm{1}_{\phi}$ be the standard indicator function which is $1$ if constraint $\phi(\vec{x})$ holds and $0$ otherwise.
The \emph{constraint satisfaction} metric measures the proportion of the (finite) test set for which the constraint holds.
\begin{definition}[\acf{CSat}]\label{def:csat}
  \begin{equation*}
    \text{CSat}(\mathcal{X}) = \frac{1}{|\mathcal{X}|} \sum_{\hat{\vec{x}} \in \mathcal{X}} \mathbbm{1}_{\forall \vec{x} \in B_\varepsilon(\hat{\vec{x}})\ldotp P(\vec{x})}
    = \frac{1}{|\mathcal{X}|} \sum_{\hat{\vec{x}} \in \mathcal{X}} \inf_{\vec{x} \in B_\varepsilon(\hat{\vec{x}})} \mathbbm{1}_{P(\vec{x})}
  \end{equation*}
\end{definition}
In contrast, \emph{constraint security} measures the proportion of inputs in the dataset such that an attack is unable to find an adversarial example for constraint $P$. In our experiments we use the \ac{PGD} attack, yielding a set (usually singleton) of adversarial examples $A(\vec{x}) \subseteq B_\varepsilon(\hat{\vec{x}})$, although in general any sufficiently strong attack can be used.
\begin{definition}[\acf{CSec}]\label{def:csec}
  \begin{equation*}
    \text{CSec}(\mathcal{X}) = \frac{1}{|\mathcal{X}|} \sum_{\hat{\vec{x}} \in \mathcal{X}} \mathbbm{1}_{\neg \exists \vec{x}\in A(\hat{\vec{x}})\ldotp \neg P(\vec{x})}
  \end{equation*}
\end{definition}
Note that since $\neg \exists \vec{x}\in A(\hat{\vec{x}})\ldotp \neg P(\vec{x}) = \forall \vec{x}\in A(\hat{\vec{x}})\ldotp P(\vec{x})$, and since $A(\vec{x}) \subseteq B_\varepsilon(\hat{\vec{x}})$, we know in general that CSat implies CSec but we cannot conclude the converse (since our attack might not find a counterexample even if it exists).

We evaluate \ac{CSec} under two different attack settings:
\begin{inparaenum}[(a)]
	\item \emph{self oracle}: adversarial samples are found using a \ac{PGD} attack with a loss built from the same differentiable logic used in training; and
	\item \emph{\ac{QLL} oracle}: adversarial samples are found using a \ac{PGD} with a \ac{QLL}-based loss for all experiments.
\end{inparaenum}
The latter is not a definite measure of constraint satisfaction (as explained above) but only provides a more uniform evaluation setting, reducing the risk that weaker attack losses overestimate constraint security by failing to find counterexamples.

Finally, \emph{constraint accuracy} estimates the probability of a random user coming across a counter-example to the constraint, usually referred as \emph{$1$ - success rate} in the robustness literature.
Let $S(\hat{\vec{x}},n)$ be a set of $n$ elements randomly uniformly sampled from $B_\varepsilon(\hat{\vec{x}})$. Then constraint accuracy is defined as:
\begin{definition}[\acf{CAcc}]\label{def:cacc}
  \begin{equation*}
    \text{CAcc}(\mathcal{X}) = \frac{1}{|\mathcal{X}|} \sum_{\hat{\vec{x}} \in \mathcal{X}} \left(\frac{1}{n} \sum_{\vec{x} \in S(\hat{\vec{x}},n)} \mathbbm{1}_{P(\vec{x})}\right)
  \end{equation*}
\end{definition}
This is the least reliable metric, as random sampling in the $m$-dimensional space is not likely to find adversarial samples.

\subsection{Datasets and Models}\label{subsec:models-datasets}
We make use of the MNIST~\citep{lecunGradientbasedLearningApplied1998} and Fashion-MNIST~\citep{xiaoFashionMNISTNovelImage2017} consisting of \num{60000} training and \num{10000} test images of size $28\times 28$, with labels in $\{0,\ldots,9\}$.
We use the standard train/test split provided by \texttt{torchvision} and a batch size of $512$.

\begin{wraptable}[8]{R}{0.5\linewidth}
    \centering
    \footnotesize
    \caption{Label distribution in the Dice dataset.}
    \label{tab:dice-label-distribution}
    \begin{tblr}{
        colspec={l*{6}{Q[c, m, wd=0.4cm]}},
        row{1}={font=\bfseries},
    }
        \toprule
        Label & 1 & 2 & 3 & 4 & 5 & 6 \\
        \midrule
        Count & 156 & 181 & 179 & 161 & 159 & 184 \\
        Frequency (\%) & 45.9 & 53.2 & 52.6 & 47.4 & 46.8 & 54.1 \\
        \bottomrule
    \end{tblr}
\end{wraptable}

The Dice dataset is a small custom multi-label dataset consisting of $340$ RGB images of size $28\times 28$.
Each image shows three visible faces of a die.
Labels are multi-hot vectors in $\{0,1\}^6$, where each ground-truth label contains exactly three positive entries.
The dataset was created by randomly rolling and photographing physical dice and manually annotating the visible faces in each image.
As a result, the label distribution is reasonably balanced across the six labels (c.f.~\cref{tab:dice-label-distribution}).
We randomly split the dataset into \qty{80}{\percent} training and \qty{20}{\percent} test data, resulting in $272$ training and $68$ test images.
Given the small training set size, we augment the training images using colour jitter to increase variability.
We use a batch size of $16$.
We formulate the Dice task as a standard multi-label classification problem rather than predicting ordered triples of visible faces, allowing us to express relationships between labels explicitly as constraints.

For MNIST and Fashion-MNIST, we use a small convolutional network with two strided $3\times 3$ convolutional layers with $4$ and $8$ channels, respectively.
The first convolutional layer is followed by a $\operatorname{ReLU}$ activation.
The resulting representation is flattened and passed through a fully-connected layer with a $\operatorname{ReLU}$ activation, followed by a final linear layer mapping to $10$ logits.

The Dice dataset uses a similar architecture, except that the first convolutional layer takes RGB images as inputs, and the final linear layer maps to $6$ logits.

\subsection{Additional Experimental Results}\label{subsec:additional-experimental-results}

\Cref{tab:main_results-long} displays the full results of the experiments we ran.

\begin{table}
  %\centering
  \tiny
  \centerline{
  \begin{talltblr}
  [
    theme=paper,
	caption={Experimental results for training and verifying models on MNIST, Fashion-MNIST, and the Dice using different constraints and differentiable logics. In each experiment, the best-performing configuration with respect to the combination of \acs{PAcc} and \acs{CSat} is highlighted in boldface.},
	label={tab:main_results-long},
    note{a}={The \enquote{unknown} column denotes verification instances that could neither be proved or disproved within a \qty{30}{\second} timeout.}
  ]
  {
	colspec={Q[l, m, cmd=\rotatebox{90}]Q[l, m, cmd=\rotatebox{90}]l*{8}{S[table-format=3.1(2.1)]}},
    rowsep=1pt,
	row{1}={guard, font=\bfseries, mode=text, cmd=},
	row{2}={guard, mode=text},
    row{9,24,34,37,38,41,42,50,54}={font=\bfseries}
  }
	\toprule
    \SetCell[r=2,c=1]{l} Constraint
    & 
    \SetCell[r=2,c=1]{l} Dataset
    &
	  \SetCell[r=2,c=1]{l} Logic
	  &
	  \SetCell[r=2,c=1]{c} {\acs{PAcc} (\unit{\percent})}
	  &
	  \SetCell[r=2,c=1]{c} {\acs{CAcc}$_{\epsilon}$ (\unit{\percent})}
	  &
	  \SetCell[r=1,c=2]{c} {\acs{CSec}$_{\epsilon}$ (\unit{\percent})}
	  &
	  &
	  \SetCell[r=1,c=2]{c} \acs{CSat}$_{\nicefrac{\epsilon}{2}}$ (\unit{\percent})
	  &
	  &
	  \SetCell[r=1,c=2]{c} \acs{CSat}$_{\epsilon}$ (\unit{\percent})
	  &
	  \\
	\cmidrule[lr]{6-7} \cmidrule[lr]{8-9} \cmidrule[lr]{10-11}
	  & & & & & self & \acs{QLL} oracle & verified & unknown\TblrNote{a} & verified & unknown\TblrNote{a}\\
	\midrule
      \SetCell[r=10,c=1]{l} \shortstack{Clothing/Footwear\\($\epsilon=0.1$)} & \SetCell[r=10,c=1]{l} \shortstack{Fashion-\\MNIST} & Baseline & 87.4+-0.5 & 99.9+-0.0 & \text{n/a} & 35.8+-11.1 & 54.0+-11.0 & 12.4+-7.4 & 19.0+-0.0 & 1.8+-2.5 \\
      & & DL2 & 86.9+-0.8 & 99.8+-0.0 & 99.8+-0.0 & 61.8+-7.8 & 85.6+-6.7 & 8.8+-5.8 & 31.0+-19.0 & 36.6+-11.8 \\
    \cmidrule[lr]{3-11}
      & & ${}_1$\acs{STL} & 87.0+-0.7 & 99.9+-0.0 & 99.8+-0.0 & 58.3+-7.3 & 85.0+-8.0 & 10.2+-7.6 & 28.4+-16.0 & 34.4+-10.5 \\
      & & ${}_2$\acs{STL} & 87.0+-0.7 & 99.9+-0.0 & 99.8+-0.0 & 57.7+-7.5 & 82.2+-8.0 & 13.4+-7.0 & 27.2+-15.0 & 35.4+-10.7 \\
      & & ${}_5$\acs{STL} & 87.2+-0.5 & 99.9+-0.0 & 99.8+-0.0 & 56.2+-6.7 & 79.2+-8.0 & 14.6+-8.1 & 25.8+-11.9 & 35.0+-5.5 \\
      & & ${}_{10}$\acs{STL} & 87.1+-0.8 & 99.8+-0.0 & 99.8+-0.0 & 56.1+-6.8 & 81.4+-7.6 & 13.4+-8.0 & 26.6+-11.7 & 34.8+-7.2 \\
    \cmidrule[lr]{3-11}
      & & ${}_1$\acs{QLL} & 84.4+-1.2 &  99.9+-0.0 &  99.4+-0.2 &  99.4+-0.2 &  99.6+-0.5 &  0.2+-0.4 &  97.4+-2.5 &  2.0+-2.4\\
      & & ${}_2$\acs{QLL} & 84.9+-1.3 & 99.9+-0.0 & 99.0+-0.4 & 99.0+-0.4 & 99.6+-0.5 & 0.2+-0.4 & 90.4+-12.6 & 8.2+-11.7 \\
      & & ${}_5$\acs{QLL} & 85.2+-1.7 & 99.9+-0.0 & 99.0+-0.3 & 99.0+-0.3 & 99.8+-0.4 & 0.0+-0.0 & 87.2+-12.3 & 12.0+-11.7 \\
      & & ${}_{10}$\acs{QLL} & 85.3+-1.0 & 99.9+-0.0 & 99.0+-0.3 & 98.9+-0.4 & 99.4+-0.9 & 0.4+-0.9 & 88.0+-11.4 & 10.4+-9.8 \\
    \midrule
      \SetCell[r=14,c=1]{l} \shortstack{Not-Both \\($\epsilon=\nicefrac{4}{255}$)} & \SetCell[r=14,c=1]{l} Dice & Baseline & 83.8+-4.7 & 96.5+-2.9 & \text{n/a} & 56.5+-8.1 & 40.9+-17.8 & 2.9+-2.3 & 12.4+-12.1 & 5.6+-2.4 \\
      && DL2 & 83.7+-1.5 & 97.4+-1.6 & 95.9+-2.4 & 65.0+-12.1 & 44.4+-17.4 & 2.1+-2.5 & 13.5+-16.2 & 4.4+-3.4 \\
    \cmidrule[lr]{3-11}
      & & Gödel & 83.8+-4.7 & 96.8+-1.9 & 95.0+-3.0 & 55.3+-7.1 & 40.9+-17.8 & 2.9+-2.3 & 12.4+-12.1 & 5.6+-2.4 \\
      & & \L ukasiewicz & 83.8+-4.7 & 96.8+-1.9 & 97.4+-2.2 & 55.3+-7.1 & 40.9+-17.8 & 2.9+-2.3 & 12.4+-12.1 & 5.6+-2.4 \\
      & & Product & 83.6+-3.9 & 97.1+-1.8 & 96.8+-2.2 & 57.6+-13.2 & 41.2+-19.4 & 1.2+-1.6 & 12.9+-13.2 & 4.4+-3.7 \\
      & & ${}_2$Yager & 82.9+-2.6 & 98.8+-1.2 & 97.9+-0.8 & 65.6+-12.3 & 45.3+-18.5 & 0.3+-0.7 & 11.8+-11.3 & 6.2+-2.4 \\
    \cmidrule[lr]{3-11}
      & & ${}_1$\acs{STL} & 83.8+-3.5 & 100.0+-0.0 & 96.5+-4.2 & 97.6+-2.2 & 90.6+-9.9 & 0.0+-0.0 & 82.4+-18.9 & 0.3+-0.7 \\
      & & ${}_2$\acs{STL} & 83.6+-1.5 & 100.0+-0.0 & 97.4+-2.4 & 98.2+-1.9 & 92.4+-4.6 & 0.0+-0.0 & 73.8+-16.8 & 4.4+-6.1 \\
      & & ${}_5$\acs{STL} & 83.8+-2.8 & 99.7+-0.7 & 97.1+-2.8 & 97.4+-2.4 & 90.0+-7.9 & 0.3+-0.7 & 73.5+-19.0 & 0.9+-1.3 \\
      & & ${}_{10}$\acs{STL} & 83.8+-2.5 & 100.0+-0.0 & 95.9+-4.3 & 97.1+-4.3 & 88.8+-7.6 & 0.9+-1.3 & 71.5+-19.4 & 2.1+-2.5 \\
    \cmidrule[lr]{3-11}
      & & ${}_1$\acs{QLL} & 84.4+-5.2 & 100.0+-0.0 & 97.9+-2.2 & 96.8+-2.6 & 91.8+-3.8 & 0.6+-0.8 & 77.9+-3.9 & 1.5+-1.5 \\
      & & ${}_2$\acs{QLL} & 84.1+-3.3 &  99.7+-0.7 &  98.8+-1.9 &  98.5+-2.5 &  93.8+-5.5 &  0.3+-0.7 &  82.9+-7.2 &  1.2+-0.7\\
      & & ${}_5$\acs{QLL} & 82.9+-3.9 & 100.0+-0.0 & 95.6+-2.9 & 95.3+-2.8 & 88.8+-6.2 & 0.9+-1.3 & 69.7+-17.4 & 2.4+-3.0 \\
      & & ${}_{10}$\acs{QLL} & 83.7+-2.5 & 99.7+-0.7 & 97.4+-2.8 & 97.4+-2.2 & 90.3+-7.4 & 0.6+-0.8 & 72.1+-11.9 & 2.6+-2.6 \\
    \midrule 
      \SetCell[r=8,c=1]{l} \shortstack{Exactly-One\\($\epsilon=\nicefrac{4}{255}$)} & \SetCell[r=8,c=1]{l} Dice & Baseline & 83.8+-4.7 & 90.6+-4.6 & \text{n/a} & 25.6+-7.2 & 27.6+-14.5 & 1.2+-1.2 & 6.2+-6.2 & 5.6+-3.0 \\
      & & DL2 & 83.7+-2.1 & 93.8+-2.4 & 90.0+-4.5 & 34.1+-13.4 & 25.3+-11.5 & 2.1+-3.0 & 3.2+-3.0 & 6.5+-6.7 \\
    \cmidrule[lr]{3-11}
      & & Gödel & 83.4+-1.6 & 94.1+-3.9 & 85.8+-4.5 & 24.5+-9.5 & 30.1+-11.4 & 3.7+-1.0 & 2.9+-4.2 & 10.3+-2.1 \\
      & & \L ukasiewicz & 84.4+-3.2 & 92.2+-5.6 & 92.2+-5.2 & 27.0+-6.6 & 31.9+-8.4 & 1.0+-0.8 & 5.9+-6.2 & 8.1+-3.1 \\
      & & Product & 84.2+-1.9 & 95.1+-2.2 & 92.2+-3.1 & 30.9+-10.6 & 30.9+-6.4 & 1.0+-0.8 & 3.4+-3.4 & 10.3+-2.9 \\
      & & ${}_2$Yager & 85.9+-2.1 & 96.1+-2.2 & 91.2+-6.7 & 34.3+-8.5 & 28.9+-11.9 & 2.0+-0.8 & 3.4+-3.1 & 6.4+-0.8 \\
    \cmidrule[lr]{3-11}
      & & ${}_5$\acs{STL} & 82.9+-3.4 & 97.4+-2.2 & 41.8+-12.7 & 44.7+-16.1 & 34.7+-16.6 & 1.8+-2.4 & 10.6+-9.1 & 4.4+-2.1 \\
      & & ${}_5$\acs{QLL} & 82.5+-4.4 &  95.9+-3.0 &  51.5+-14.4 &  49.7+-15.3 &  42.4+-14.7 &  1.8+-1.2 &  14.7+-10.3 &  5.9+-5.4\\
     \midrule
      \SetCell[r=8,c=1]{l} \shortstack{Strong classification\\robustness ($\epsilon=0.1$)} &\SetCell[r=4,c=1]{l} MNIST & Baseline & 97.6+-0.4 & 99.9+-0.1 & \text{n/a} & 19.3+-18.1 & 66.6+-7.5 & 11.6+-13.7 & 2.0+-1.9 & 14.2+-18.1 \\
      & & DL2 & 97.6+-0.4 & 100.0+-0.0 & 100.0+-0.0 & 62.7+-19.5 & 95.6+-2.7 & 3.0+-2.4 & 24.2+-10.3 & 29.2+-19.8 \\
      & & ${}_5$\acs{STL} & 97.5+-0.3 &  100.0+-0.0 &  100.0+-0.0 &  100.0+-0.0 &  100.0+-0.0 &  0.0+-0.0 &  93.4+-9.4 &  6.6+-9.4\\
      & & ${}_5$\acs{QLL} & 97.5+-0.3 &  100.0+-0.0 &  100.0+-0.0 &  100.0+-0.0 &  100.0+-0.0 &  0.0+-0.0 &  93.4+-9.4 &  6.6+-9.4\\
    \cmidrule[lr]{2-11}
      & \SetCell[r=4,c=1]{l} \shortstack{Fashion-\\MNIST} & Baseline & 87.4+-0.5 & 99.6+-0.2 & \text{n/a} & 1.5+-1.5 & 41.0+-15.0 & 9.8+-7.8 & 0.0+-0.0 & 0.6+-0.9 \\
      & & DL2 & 87.0+-0.5 & 100.0+-0.0 & 100.0+-0.0 & 39.4+-16.0 & 85.0+-7.0 & 6.2+-3.9 & 15.2+-9.3 & 17.8+-6.1 \\
      & & ${}_5$\acs{STL} & 86.7+-0.6 &  100.0+-0.0 &  100.0+-0.0 &  100.0+-0.0 &  100.0+-0.0 &  0.0+-0.0 &  99.0+-1.0 &  1.0+-1.0\\
      & & ${}_5$\acs{QLL} & 86.7+-0.6 &  100.0+-0.0 &  100.0+-0.0 &  100.0+-0.0 &  100.0+-0.0 &  0.0+-0.0 &  99.0+-1.0 &  1.0+-1.0\\
    \midrule
	  \SetCell[r=12,c=1]{l} \shortstack{Classification robustness\\($\epsilon=0.2$)} & \SetCell[r=8,c=1]{l} MNIST & Baseline & 97.6+-0.4 & 94.0+-1.9 & \text{n/a} & 0.0+-0.0 & 0.0+-0.0 & 13.8+-23.1 & 0.0+-0.0 & 0.0+-0.0 \\
	  & & DL2 & 97.1+-1.2 & 84.5+-17.5 & 70.4+-34.0 & 23.1+-20.3 & 38.6+-40.9 & 12.2+-14.7 & 10.4+-21.6 & 26.4+-32.8 \\
    \cmidrule[lr]{3-11}
      & & Gödel & 97.5+-0.4 & 95.3+-1.1 & 94.7+-1.6 & 0.3+-0.3 & 0.0+-0.0 & 15.8+-16.8 & 0.0+-0.0 & 0.4+-0.5 \\
      & & \L ukasiewicz & 97.6+-0.3 & 94.6+-1.3 & 92.3+-5.1 & 0.2+-0.1 & 0.0+-0.0 & 13.2+-16.8 & 0.0+-0.0 & 0.2+-0.4 \\
      & & Product & 97.6+-0.4 & 95.3+-1.0 & 94.3+-3.0 & 0.3+-0.2 & 0.0+-0.0 & 12.0+-5.9 & 0.0+-0.0 & 0.5+-0.6 \\
      & & ${}_2$Yager & 97.5+-0.6 & 94.8+-1.2 & 93.3+-2.7 & 0.2+-0.3 & 0.0+-0.0 & 4.8+-3.3 & 0.0+-0.0 & 0.2+-0.5 \\
    \cmidrule[lr]{3-11}
	  & & ${}_5$\acs{STL} & 97.7+-0.5 & 85.0+-18.4 & 68.5+-34.9 & 8.2+-3.4 & 1.4+-2.6 & 23.4+-19.7 & 0.0+-0.0 & 4.0+-6.2 \\
	   & & ${}_5$\acs{QLL} & 96.1+-0.8 &  95.9+-0.8 &  67.7+-4.0 &  67.8+-4.0 &  88.2+-2.7 &  0.8+-1.1 &  46.0+-23.2 &  30.0+-28.8\\
    \cmidrule[lr]{2-11}
      & \SetCell[r=4,c=1]{l} \shortstack{Fashion-\\MNIST} & Baseline & 87.4+-0.5 & 82.0+-0.8 & \text{n/a} & 0.1+-0.1 & 0.0+-0.0 & 0.4+-0.9 & 0.0+-0.0 & 0.0+-0.0 \\
	   & & DL2 & 85.0+-1.7 & 74.3+-6.0 & 62.1+-11.3 & 7.8+-3.9 & 6.2+-13.3 & 34.4+-16.8 & 0.0+-0.0 & 12.0+-12.1 \\
	   & & ${}_5$\acs{STL} & 84.6+-1.1 & 80.5+-3.4 & 70.1+-7.7 & 3.3+-3.2 & 4.4+-8.2 & 22.0+-14.6 & 0.0+-0.0 & 1.6+-1.7 \\
	   & & ${}_5$\acs{QLL} & 79.5+-1.0 & 79.4+-1.0 & 40.1+-2.1 & 40.1+-2.0 & 56.0+-3.1 & 6.4+-2.2 & 23.2+-9.6 & 36.8+-17.7 \\
	\bottomrule
  \end{talltblr}
  }
\end{table}

\end{document}